\documentclass[11pt]{article}
\usepackage{amsmath,amssymb,amsthm} 
\usepackage{bbm}
\usepackage{stmaryrd}
\usepackage{paralist}
\usepackage{hyperref}
\usepackage{color}

\usepackage{multirow}
\setdefaultenum{(i)}{}{}{}
\usepackage[margin=1in]{geometry}
\usepackage[pdftex]{graphicx}
\usepackage{subfigure}
\usepackage{cases}
\theoremstyle{plain}
\newtheorem{mythm}{Theorem} \numberwithin{mythm}{section}
\newtheorem{myprop}[mythm]{Proposition}
\newtheorem{mylemma}[mythm]{Lemma}
\newtheorem{mydef}[mythm]{Definition}
\newtheorem{myrek}[mythm]{Remark}

\DeclareMathAlphabet\scr{U}{scr}{m}{n}
\SetMathAlphabet\scr{bold}{U}{scr}{b}{n}
  \DeclareFontFamily{U}{scr}{\skewchar\font'177}%
  \DeclareFontShape{U}{scr}{m}{n}{<-6>rsfs5<6-8>rsfs7<8->rsfs10}{}%
  \DeclareFontShape{U}{scr}{b}{n}{<-6>rsfs5<6-8>rsfs7<8->rsfs10}{}%

\numberwithin{equation}{section}

\renewenvironment{thebibliography}[1]{%
\begin{oldthebibliography}{#1}%
\setlength{\baselineskip}{.9em}
\linespread{.9}
\small
\setlength{\parskip}{0ex}%
\setlength{\itemsep}{.1em}%
}%
{%
\end{oldthebibliography}%
}

\DeclareMathOperator{\esr}{ESR}
\DeclareMathOperator{\rr}{\mathbb{R}}

\begin{document}
\title{Rebalancing with Linear and Quadratic Costs\footnote{The authors are grateful to Paolo Guasoni for fruitful discussions, and to two anonymous referees for their pertinent remarks.}}
\author{
Ren Liu\thanks{ETH Z\"urich, Departement Mathematik, R\"amistrasse 101, CH-8092, Z\"urich, Switzerland, email:
\texttt{ren.liu@math.ethz.ch}.}
\and
Johannes Muhle-Karbe\thanks{Carnegie Mellon University, Department of Mathematical Sciences, 5000 Forbes Avenue, Pittsburgh, PA 15213, USA, email:
\texttt{johannesmk@cmu.edu}.}
\and
Marko H. Weber\thanks{Columbia University, Department of Industrial Engineering and Operations Research, 535G S.\ W.\ Mudd Building, New York, NY 10027, USA  email:
\texttt{mhw2146@columbia.edu}.}}

\date{\today}
\pagestyle{plain}
\maketitle
\begin{abstract}
We consider a financial market consisting of one safe and one risky asset, which offer constant investment opportunities. Taking into account both proportional transaction costs and linear price impact, we derive optimal rebalancing policies for representative investors with constant relative risk aversion and a long horizon.
\end{abstract}

\noindent\textbf{Mathematics Subject Classification: (2010)}: 91G10, 91G80 

\noindent\textbf{JEL Classification:} G11, G12

\noindent\textbf{Keywords:} price impact, transaction costs, portfolio choice, long-run.

\section{Introduction}

Proportional transaction costs are ubiquitous even in the most liquid financial markets in the form of bid-spreads. For large institutional investors, the price impact of their trades also is a key concern.\footnote{For small private investors, fixed commissions, levied on each trade regardless of its size, are also crucial (cf., e.g., \cite{korn.98} and the references therein). Yet, for large portfolios, their influence becomes negligible \cite{altarovici.al.13}, and we disregard them in the present study.} As a result, both frictions have generated voluminous literatures that analyze how to balance the gains and costs of portfolio rebalancing optimally. Bid-ask spreads lead to trading costs \emph{linear} in the amounts transacted. Then, it is optimal to refrain from trading while the uncontrolled portfolio lies inside some ``no-trade region'' around the frictionless optimum; once its boundaries are breached, one performs the minimal amount of rebalancing necessary to remain inside.\footnote{See \cite{magill.constantinides.76,constantinides.86,davis.norman.90,dumas.luciano.91,shreve.soner.94,dai.yi.09} for models with constant investment opportunities, \cite{shreve.soner.94,janecek.shreve.04,bichuch.12,gerhold.al.14} for corresponding explicit asymptotic formulas, and \cite{martin.12,soner.touzi.13,possamai.al.13,kallsen.muhlekarbe.13,kallsen.li.13} for recent extensions to more general settings.} In contrast, linear price impact leads to \emph{quadratic} trading costs, which are less severe for small trades but become prohibitively expensive for larger orders. As a result, optimal policies typically prescribe rebalancing at all times but at a finite absolutely continuous rate, in contrast to the singular controls used with proportional transaction costs.\footnote{Price impact has been studied extensively in the optimal execution literature, see, e.g., \cite{bertsimas.lo.98,almgren.chriss.01,schied.schoeneborn.09}.  Recently, increasing attention has also been devoted to its influence on dynamic portfolio choice \cite{garleanu.pedersen.13a,garleanu.pedersen.13b,almgren.li.11,guasoni.weber.13,dufresne.al.12,moreau.al.14}.} 

All of the extant literature studies either linear or quadratic trading costs. This paper fills this gap by analyzing the \emph{joint} impact of proportional transaction costs and linear price impact on portfolio rebalancing. 

To make the model tractable, we focus (as in Gerhold et al.~\cite{gerhold.al.14} resp.~Guasoni and Weber~\cite{guasoni.weber.13} for proportional resp.~quadratic costs treated separately) on a single risky asset with constant investment opportunities, and a representative agent with constant relative risk aversion and a long horizon. The representative investor's wealth serves as a proxy for a multiple of market capitalization, and prize impact is assumed to be inversely proportional to latter. That is, a trade of a given size has a smaller impact as markets grow, compare the discussion in Guasoni and Weber~\cite{guasoni.weber.13}. The present study builds on the results of Guasoni and Weber~\cite{guasoni.weber.13} by adding a bid-ask spread. This price impact model leads to stable long-term behavior. The most popular alternative -- constant price impact as in \cite{bertsimas.lo.98,madhavan.00,almgren.chriss.01} -- is well-suited to the short-horizons optimal execution problems these authors have in mind. However, it leads to degenerate results in the long run, as growing fund values eventually make rebalancing prohibitively expensive. To bridge the gap between the different models, we argue that, asymptotically for small costs, our results formally extend to rather general specifications of the price impact function.

The optimal policy in the presence of both frictions turns out to be of the following form. As with proportional transaction costs, there is a no-trade region, where it is optimal to simply hold the current portfolio. Once its boundaries are breached, price impact rules out singular controls; hence, one instead starts trading at some finite rate so as to steer the portfolio back to the no-trade region. Since this policy does not allow to keep the portfolio uniformly close to the frictionless target, trading starts earlier than 
in a model with only proportional costs, i.e., the width of the no-trade region is decreased by the additional price impact. On the other hand, since there are now two frictions contributing to the total rebalancing cost, the total trading rate is always lower than in a model that only takes into account price impact. We prove a rigorous verification theorem that identifies the trading boundaries, the trading rate, and the associated welfare through the solution of a nonlinear free-boundary value problem that can be solved numerically. To ease implementation, we also present asymptotics for small linear and quadratic costs, which reduce the computation of the optimal policy and welfare to finding the root of a scalar function.\footnote{This is similar in complexity to the asymptotics for small fixed and proportional costs studied by \cite{korn.98}.} With two competing frictions, it is particularly important to assess the quality of the asymptotic approximations, since it is not clear a priori whether the matched rescaling leads to 
accurate results. Our exact formulas allow to do this, and show that the small-cost approximations perform very well. 

A formal extension of our asymptotic results suggests that, similarly as for proportional costs~\cite{kallsen.muhlekarbe.13} and price impact~\cite{moreau.al.14} considered separately, the optimal policy is robust to the form of the trading cost. To wit, in Section \ref{ss:robust}, we consider general price impact functions depending on the large investor's wealth and an additional exogenous state variable (say, the wealth of other traders) in an arbitrary manner. Then, the same asymptotic trading rate as in our baseline model obtains, substituting the current value of the trading cost at each point in time.

 The remainder of the paper is organized as follows. Section \ref{sec:model} describes the model. Our main results are presented in Section~\ref{sec:main}. Subsequently, we illustrate them with some numerical examples. Section~\ref{heuristic} contains a heuristic derivation of our main results; our asymptotics are presented in Section~\ref{sec:asymptotics}. For better readability, all proofs are delegated to Section~\ref{sec:proofs}.

\section{Model}\label{sec:model}
Consider a market consisting of one safe asset normalized to one,\footnote{That is, the interest rate $r$ is set equal to zero.} and one risky asset, whose mid price $S_t$ follows geometric Brownian motion:
\begin{equation}
\frac{d S_t}{S_t}=\mu d t+\sigma d W_t.
\end{equation}
Here, $(W_t)_{t \geq 0}$ is a standard one-dimensional Brownian motion, $\mu >0$ is the expected excess return, and $\sigma > 0$ is the volatility. Trades are not settled at the idealized best quote $S_t$. Instead, sales only earn lower bid prices, whereas purchases are charged higher ask prices. Moreover, trading large positions quickly  moves prices further in an adverse direction. To wit, the average execution price for trading $\Delta\theta$ shares over a time interval $\Delta t$ is 
\begin{equation}\label{modelassump}
S_t\left(1+\varepsilon\mathrm{sgn}(\Delta\theta)+\lambda\frac{S_t\Delta\theta }{X_t\Delta t}\right).
\end{equation}
Here, the first term corresponds to a relative bid-ask spread $\varepsilon$, i.e., a higher ask price $(1+\varepsilon)S_t$ for purchases and a lower bid price $(1-\varepsilon)S_t$ for sales, respectively. The second term describes the additional (relative) price impact of large trades executed quickly. This price impact is proportional to the monetary trading rate $S_t\Delta\theta/\Delta t$, and inversely proportional to market capitalization, which is proxied by the representative investor's wealth $X_t$.\footnote{See \cite{guasoni.weber.13} for more details on this price impact model and the related literature.}  The constant of proportionality $\lambda$ in turn quantifies the market's limited liquidity; put differently, $1/\lambda$ measures market ``depth''. For $\varepsilon, \lambda \to 0$, one recovers the classical frictionless case, where arbitrary amounts $\Delta\theta$ can be purchased or sold at the mid price for $S_t \Delta\theta $. Nontrivial bid-ask spreads ($\varepsilon>0$) and finite market 
depth ($\lambda>0$) lead to additional linear and quadratic trading costs, respectively. Specifically, with both frictions, the execution cost of trading $\Delta\theta$ shares over a time interval $\Delta t$ is given by\footnote{Note the price impact is purely temporary in our model, in that no trade influences the subsequent ones. There is a large literature on optimal execution with persistent price impact, that only wears off gradually after the completion of each trade (cf.\ \cite{obizhaeva.wang.13} as well as many more recent studies). Since we are working on a much longer time scale than in this literature, we abstract from this issue and instead suppose that the temporary and persistent impact costs generated by various ``sub-trades'' on a finer ``execution time-grid'' are all aggregated into our price impact cost. Indeed, suppose each infinitesimal sub-trade is executed in the setting of Obizhaeva and Wang~\cite{obizhaeva.wang.13}. Then, the expected execution costs are of the same linear-quadratic form as in our model.}
\begin{equation}\label{defpriceimpact}
S_t \Delta\theta+\varepsilon S_t \left|\frac{\Delta\theta}{\Delta t}\right|\Delta t+ \lambda \frac{S_t^2\Delta\theta^2}{X_t \Delta t^2}\Delta t.
\end{equation}

 For tractability, we now pass to the continuous-time limit. Denote by $\theta_t$ the number of risky shares the investor holds at time $t$, and replace $\Delta \theta/\Delta t$ in~\eqref{defpriceimpact} with $\dot{\theta}_t := \lim_{h \downarrow 0} \frac{\theta_{t+h}-\theta_t}{h} $. Then, the investor's cash position $C_t = X_t- S_t \theta_t$ evolves as
\begin{equation*}
dC_t = - S_t d\theta_t -\varepsilon S_t |\dot{\theta}_t|dt - \lambda \frac{S_t^2 \dot \theta^2_t}{X_t} dt.
\end{equation*}
Write $u_t := \dot \theta_t S_t/X_t$ for the \emph{wealth turnover} at time $t$. With this notation, a straightforward application of It\^o's formula (compare~\cite[Lemma A.1]{guasoni.weber.13}) shows that the corresponding wealth process $X_t:= \theta_t S_t + C_t$ and risky weight $Y_t := \theta_t S_t/X_t$ have the following dynamics:
\begin{align}
\frac{d X_t}{X_t} &= Y_t (\mu dt + \sigma dW_t)- \varepsilon |u_t| dt -\lambda u_t^2 dt,\label{dynamicsX}\\
dY_t &= (Y_t(1-Y_t)(\mu-Y_t \sigma^2)+u_t+ \varepsilon |u_t| Y_t + \lambda Y_t u_t^2)dt+ Y_t(1-Y_t) \sigma dW_t.\label{dynamicsY}
\end{align}

As without proportional transaction costs \cite{guasoni.weber.13}, linear price impact implies that the risky weight is no longer a control variable that can be specified freely by the investor. Instead, it becomes a state variable, for which only the drift rate can be influenced by applying the control $u$. To make this precise, fix a filtered probability space $(\Omega, \mathcal{F},(\mathcal{F})_{t\geq 0},\mathbb{P})$ supporting a Brownian motion $(W_t)_{t\geq 0}$, where $\mathcal{F}_t$ is the augmentation of the filtration generated by $W$. We then define strategies in terms of the control variable $(u_t)_{t\geq 0}$. To rule out doubling strategies, we focus on \emph{admissible} strategies with positive wealth process $X^u$:

\begin{mydef}
An \emph{admissible} strategy is an adapted process $(u_t)_{t\geq 0}$, which is square-integrable (i.e., $\int_0^T u_t^2 dt < \infty$ a.s.\ for all $T>0$) and such that~(\ref{dynamicsX}) has a unique strong solution on $[0,\infty)$ for any $Y_0 \in [0,1]$. For any such admissible strategy, the corresponding wealth process is\footnote{If $y_\ast:=\mu/\gamma \sigma^2 \in (0,1)$, i.e., the frictionless target portfolio prescribes neither leverage nor shortselling, then Lemma~\ref{ybound} shows that $Y_t$ takes values in $[0,1]$ almost surely for all $t$. In particular, $\int_0^T Y_t^2 dt < \infty$ so that the process $X^u$ is well defined in this case.}
\begin{equation*}
X_t^u = X_0 \exp{\left(\int_0^T \left(\mu Y_t -\frac{\sigma^2}{2} Y_t^2 -  \varepsilon |u_t| -\lambda u_t^2\right) dt + \int_0^T \sigma Y_t dW_t\right)}.
\end{equation*}
\end{mydef}

As in \cite{dumas.luciano.91,grossman.vila.92,grossman.zhou.93}, the representative investor has constant relative risk aversion and maximizes the growth rate of her expected utility from terminal wealth over a long horizon. Put differently, she maximizes the ``equivalent safe rate'', for which a full safe investment yields the same utility as investing optimally in the original market:

\begin{mydef}\label{longrun}
An admissible strategy $(u_t)_{t\geq 0}$ is called \emph{long-run optimal}, if it maximizes the \emph{equivalent safe rate}
\begin{equation}\label{deflongrun}
\esr_\gamma(u):=\liminf_{T\rightarrow \infty}\frac{1}{T}\log{\mathbb{E}\left[(X_T^{u})^{1-\gamma}\right]^{\frac{1}{1-\gamma}}}
\end{equation}
over all admissible strategies, where $0<\gamma\not=1$ denotes the investor's relative risk aversion.
\end{mydef}

\section{Main Results}\label{sec:main}

Our main results can be summarized as follows:
\begin{mythm}\label{mainresult}
An investor with constant relative risk aversion $0< \gamma \neq 1$ trades to maximize the equivalent safe rate~\eqref{deflongrun}, in the presence of a nontrivial bid-ask spread $\varepsilon$ and finite market depth~$1/\lambda$. Then, if $y_\ast:=\mu/\gamma\sigma^2 \in (0,1)$ and the trading costs $\varepsilon, \lambda$ are sufficiently small:
\begin{enumerate}
\item[i)] There exist constants $\beta \in [\max\{0,\mu-\gamma\sigma^2/2\},\mu^2/2\gamma\sigma^2]$ and $0 \leq y_- \leq y_+ \leq 1$, as well as a $C^1$-function $q: [0,1] \to \mathbb{R}$ which solve the ODE
\begin{equation}\label{eq:odeq}
\begin{split}
0&=-\beta + \mu y -\frac{\gamma \sigma^2}{2} y^2+ y(1-y) (\mu-\gamma \sigma^2 y) q + \frac{\sigma^2}{2}y^2(1-y)^2(q'+(1-\gamma)q^2)\\
&\qquad+\begin{cases}
\frac{1}{4 \lambda}\frac{(q-\varepsilon(1-y q))^2}{1-y q}, &\text{if } y\in [0,y_-],\\
 0, &\text{if } y\in[y_-,y_+],\\
\frac{1}{4 \lambda}\frac{(q+\varepsilon(1-y q))^2}{1-y q},  &\text{if } y\in [y_+,1],
\end{cases}
\end{split}
\end{equation}
with boundary conditions
\begin{align}
q(0^+)&= \varepsilon + 2 \sqrt{\lambda \beta},\label{abelini0}\\
\label{abelini}
q(1^-)&= \frac{\lambda d -\varepsilon (1-\varepsilon)- \sqrt{\lambda d (\lambda d -2 + 2 \varepsilon)}}{(1-\varepsilon)^2}, \quad \mbox{where } d := -\gamma \sigma^2- 2 \beta+ 2 \mu,\\
q(y_-) &= \frac{\varepsilon}{1+\varepsilon y_-},\label{abelend1}\\
q(y_+) &= \frac{-\varepsilon}{1-\varepsilon y_+}.\label{abelend}
\end{align}
\item[ii)]
A long-run optimal strategy $\hat{u}$ is to remain inactive while the corresponding risky weight lies in the no-trade region $[y_-,y_+]$, and to rebalance at the following wealth turnover rate if it does not:
\begin{equation*}
\hat{u}(y)= 
\begin{cases}
\frac{1}{2 \lambda}\left(\frac{q(y)}{1-y q(y)}-\varepsilon\right) \geq 0, &\text{if } y \in [0,y_-],\\
\frac{1}{2 \lambda}\left(\frac{q(y)}{1-y q(y)}+\varepsilon\right) \leq 0, &\text{if } y \in [y_+,1].\\
\end{cases}
\end{equation*}  
\item[iii)] The maximal equivalent safe rate is $\beta$.
\end{enumerate}
\end{mythm}

The constant $y_\ast=\mu/\gamma \sigma^2$ is the risky weight without frictions \cite{merton.69}. Thus, $y_\ast \in (0,1)$ means that the frictionless optimal strategy neither shorts nor levers the risky asset. As shown by Guasoni and Weber~\cite[Theorem 2.3]{guasoni.weber.13}, levered or short positions cannot be admissible with linear price impact for risk-averse investors, because they cannot be liquidated quickly enough to offset unfavorable diffusive price moves. This is only exacerbated by the additional linear trading cost. Hence, buy-and-hold strategies are optimal for $y_\ast \notin (0,1)$, as in \cite[Theorem 2.3]{guasoni.weber.13}:

\begin{myprop}\label{prop:hold}
Under the assumptions of Theorem \ref{mainresult}:
\begin{itemize}
\item [(i)] If $\mu/\gamma \sigma^2 \leq 0$, then $Y_t=0$ and $\hat{u}_t=0$ for all $t$ is long-run optimal, and $\esr_\gamma(\hat{u})=0.$
\item [(ii)] If $\mu/\gamma \sigma^2 \geq 1$, then $Y_t=1$ and $\hat{u}_t=0$ for all $t$ is long-run optimal, and $\esr_\gamma(\hat{u})=\mu-\gamma\sigma^2/2.$
\end{itemize}
\end{myprop}

The objective function \eqref{deflongrun} uses paper wealth rather than the liquidation value of the portfolio. This seems more reasonable for long-run investments meant to run indefinitely, such as trust funds and university endowments. However, similarly as in~\cite[Lemma 2.4]{guasoni.weber.13} for purely quadratic costs, it can be shown that this choice is of little consequence.\footnote{In the proof of~\cite[Lemma 2.4]{guasoni.weber.13}, use (\ref{dynamicsX}-\ref{dynamicsY}) instead of~\cite[Equation (2.4) and (2.5)]{guasoni.weber.13}. The rest of the proof carries through unchanged.} To wit, assuming a constant best quote (which is justified if the portfolio is sold quickly), a policy of selling at a constant turnover rate completes the liquidation within a short period of time and at a small fraction of portfolio value. Formally:

\begin{mylemma}\label{liquidationvalue}
Let $S_t \equiv S$ be constant for $t \geq T$. Then, the liquidation time $L(u)= \inf \{t\geq 0: \theta_{T+t}=0$\} of the constant selling policy $u_t \equiv u <0,$ equals $$ L(u) = -\frac{\log\left(1+(\varepsilon |u|+ \lambda u^2)\frac{Y_T}{u} \right) }{\varepsilon |u|+ \lambda u^2} \sim -\frac{Y_T}{u}.$$ The corresponding relative liquidation cost is $$\frac{X_T-X_{T+L(u)}}{X_T}= \varepsilon Y_T-\lambda u Y_T.$$
\end{mylemma}

For instance, if\footnote{As shown in~\cite[Equation (2.12)]{guasoni.weber.13} the optimal turnover $u^{0,\lambda}$ without linear costs admits the following asymptotic expansion:
$$u^{0,\lambda}(y)= \sigma (\gamma/2)^{1/2} (y_*-y)\lambda^{-1/2}+ o(\lambda^{-1/2}).$$} $u=-\lambda^{-1/2}$ the liquidation time is less than $\lambda^{1/2}$ years, since $Y_T\in [0,1]$. Furthermore, the corresponding liquidation cost is less than $(\varepsilon+ \lambda^{1/2})$ times the terminal wealth. Using the estimation interval of $[10^{-3},10^{-7}]$ for $\lambda$ (cf.~\cite[Section 3.1]{guasoni.weber.13}) and assuming a liquid stock with a bid-ask spread of $10$ basis points yields a liquidation time between $0.08$ and $7.91$ days and a relative liquidation cost between $0.13 \%$ and $3.26 \%$. As the horizon increases, the impact of these small costs on the equivalent safe rate vanishes. Moreover, the short liquidation time supports the constant best quote assumption.

\section{Numerical Examples}\label{sec:numerics}

In this section, we investigate the properties of the optimal rebalancing policy from Theorem~\ref{mainresult} in some numerical examples. This also allows us to assess the quality of the asymptotics established in Section~\ref{sec:asymptotics}, which turns out to be excellent.

\begin{figure}
\subfigure{\includegraphics[width=0.49\textwidth]{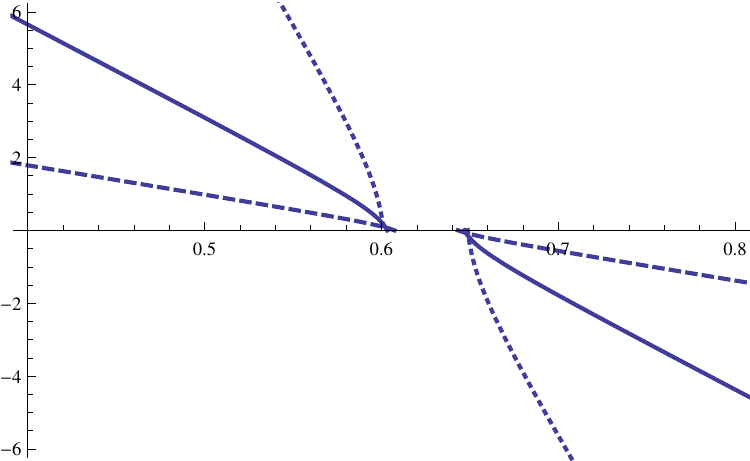}}\hfill
\subfigure{\includegraphics[width=0.49\textwidth]{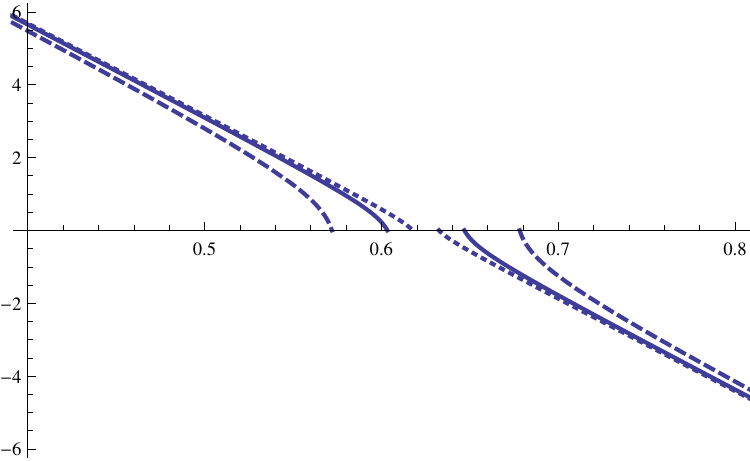}}
\caption{left panel: optimal wealth turnover $\hat{u}$ (Dotted: $\lambda=0.001\%$, Solid: $\lambda=0.01\%$, Dashed: $\lambda=0.1\%$) against the risky weight $Y$ with $\varepsilon=0.1\%$ fixed. Right panel: wealth turnover $\hat{u}$ (Dotted: $\varepsilon=0.01\%$, Solid $\varepsilon=0.1\%$, Dashed: $\varepsilon=1\%$) against the risky weight $Y$ with $\lambda=0.01\%$ fixed. Model parameters are $\mu = 8\%$, $\sigma = 16\%$, and $\gamma = 5$. }
\label{fig:tradingrates}
\end{figure}

Figure \ref{fig:tradingrates} displays the optimal policy for various trading costs. The left panel shows how the turnover rate and the corresponding no-trade region depend on the price impact parameter $\lambda$. As the latter decreases, turnover quickly increases near the boundary of the no-trade region, converging to the singular controls (``pushing at an infinite rate'') applied there with only proportional costs. For higher price impact costs, the optimal trading rate is almost linear in the deviation from the trading boundaries. Moreover, the width of the no-trade region decreases in this case, as investors start trading earlier to compensate for the slower trading rate at the boundary. However, the size of this effect is quite small, i.e., the width of the optimal no-trade region is relatively insensitive to the quadratic costs. 

The right panel in Figure \ref{fig:tradingrates} plots the trading rate for different widths $\varepsilon$ of the bid-ask spread. As the latter decreases, the no-trade region shrinks to zero and the optimal policy converges to the one with price impact only, i.e., rebalancing at a rate essentially proportional to the deviation from the frictionless Merton portfolio \cite{guasoni.weber.13}. For larger spreads, the no-trade regions widens quickly, and the optimal rebalancing rate increases much faster near the trading boundaries than further away from these.\footnote{There, it grows according to the asymptotic formula~\eqref{eq:asymp2} corresponding to a model with only quadratic costs.}

In summary, the optimal policy prescribes to i) start trading earlier than with only proportional costs, and ii) rebalance slower than with only price impact. The turnover rate increases faster near the trading boundaries; further way from these, it approaches its counterpart for only quadratic costs. 

As a complement, the quality of the small-cost asymptotics derived in Section \ref{sec:asymptotics} is assessed in Figure~\ref{fig:asymptotics}. There, we compare the optimal turnover rate to its asymptotic expansion~\eqref{eq:hatu} for two combinations of trading costs. Even for unrealistically large frictions, the approximations provide an excellent fit. Hence, the computational load can be eased to finding the root of a single scalar function with little loss in accuracy. 

\begin{figure}
\subfigure{\includegraphics[width=0.49\textwidth]{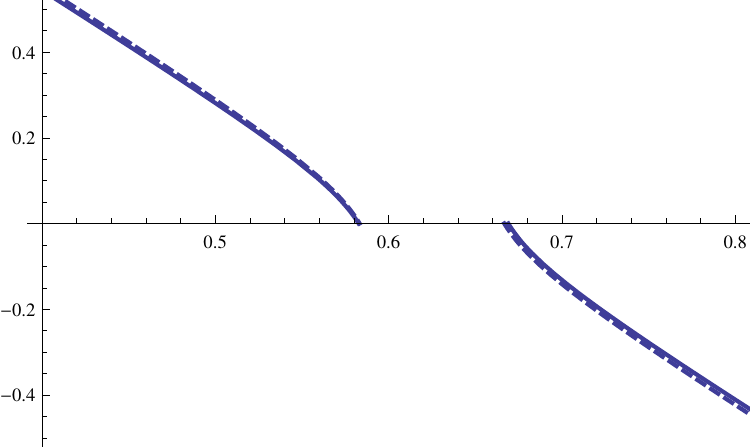}}\hfill
\subfigure{\includegraphics[width=0.49\textwidth]{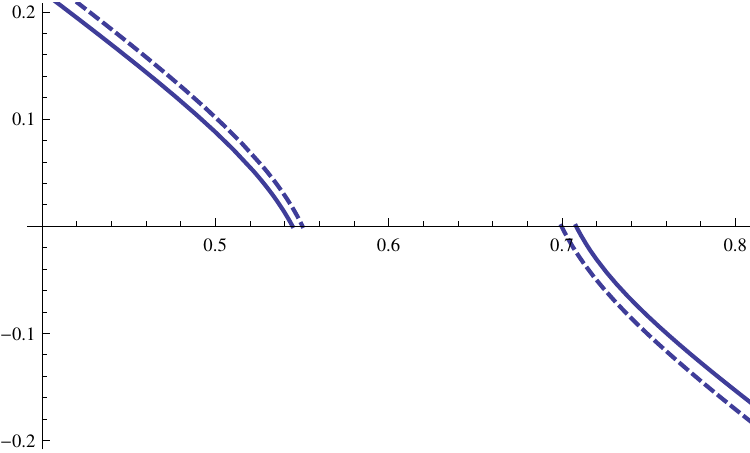}}
\caption{optimal wealth turnover $\hat{u}$ from Theorem \ref{mainresult} (solid) and its asymptotic expansion \eqref{eq:hatu} (dashed), for $\varepsilon=\lambda=1\%$ (left panel) and $\varepsilon=\lambda=5\%$ (right panel). Model parameters are $\mu = 8\%$, $\sigma = 16\%$, and $\gamma = 5$. }
\label{fig:asymptotics}
\end{figure}

%
%
%

\section{Heuristics}\label{heuristic}

In this section, we use arguments from stochastic control to heuristically derive a candidate solution for the long-run problem~\eqref{deflongrun}. To this end, consider the maximization of expected power utility $U(x) = x^{1-\gamma}/(1-\gamma)$ from terminal wealth at time $T>0$. Denote by $V(t,X_t,Y_t)$ the corresponding value function, which is assumed to depend on the current wealth $X_t$, the current risky weight $Y_t$, and time $t$. For any given strategy $u$, It\^o's formula yields:
\begin{align*}
d V(t, X_t, Y_t) = & V_t dt + V_x d X_t + V_y d Y_t + \frac{V_{xx}}{2} d \langle X \rangle _t+ \frac{V_{yy}}{2}d \langle Y \rangle _t + V_{xy}d \langle X,Y \rangle_t\\
 = & V_t dt + V_x(\mu X_t Y_t -\varepsilon X_t |u_t| -\lambda X_t u_t^2) dt + V_x X_t Y_t \sigma dW_t\\
  &+ V_y (Y_t(1-Y_t)(\mu-Y_t\sigma^2)+u_t +\varepsilon Y_t |u_t|+ \lambda Y_t u_t^2) dt + V_y Y_t(1-Y_t) \sigma d W_t\\
  &+\left(\frac{\sigma^2}{2}V_{xx}X_t^2 Y_t^2+\frac{\sigma^2}{2} V_{yy}Y_t^2(1-Y_t)^2 + \sigma^2 V_{xy}X_t Y_t^2(1-Y_t)\right)dt.
\end{align*}
By the martingale optimality principle of stochastic control, the value function $V(t,X_t,Y_t)$ must be a supermartingale for any admissible strategy, and a martingale for the optimal one. That is, the drift of $V(t,X_t,Y_t)$ cannot be positive and must become zero for the optimizer. This leads to the Hamilton-Jacobi-Bellman (henceforth HJB) equation:
\begin{align*}
0 =& V_t + y(1-y)(\mu-\sigma^2 y)V_y + \mu x y V_x + \frac{\sigma^2 y^2}{2}(x^2 V_{xx}+(1-y)^2 V_{yy}+2 x (1-y) V_{xy})\\
  & + \max_{u}{\left(-\lambda x u^2 V_x-\varepsilon x |u| V_x+ V_y (u+\varepsilon y |u|+ \lambda y u^2)\right)}.
\end{align*}
The homotheticity $U(x)=x^{1-\gamma}U(1)$ of the power utility function and the conjecture that -- in the long run -- utility should grow at a constant exponential rate motivate the following ansatz for the long-run value function:
\begin{equation}\label{reducedvaluefunction}
V(t,x,y) = \frac{x^{1-\gamma}}{1-\gamma}e^{(1-\gamma)(\beta (T-t)+\int_{p}^y q(z) dz)}.
\end{equation}
Note that the function $q$ is defined up to some arbitrary $p$. This definition of $V$ leads to the long-run version of the HJB-equation:
\begin{align}
0 =& -\beta + \mu y -\frac{\gamma \sigma^2}{2} y^2+ y(1-y)(\mu-\gamma \sigma^2 y)q + \frac{\sigma^2}{2}y^2(1-y)^2(q'+(1-\gamma)q^2)\nonumber\\
  & +\max_{u}(-\lambda u^2 -\varepsilon |u| + (u +\varepsilon |u| y + \lambda y u^2) q).
\end{align}
Decomposing wealth turnover $u$ into purchase and sale turnover, i.e., $u = u^{+}-u^{-}$, the HJB equation reduces to
\begin{align}
0 =& -\beta + \mu y -\frac{\gamma \sigma^2}{2} y^2+ y(1-y)(\mu-\gamma \sigma^2 y)q + \frac{\sigma^2}{2}y^2(1-y)^2(q'+(1-\gamma)q^2)\nonumber\\
   &+\max_{u^{+} \geq 0}(-\lambda (u^{+})^2 -\varepsilon u^{+} + (u^{+} +\varepsilon u^{+} y + \lambda y (u^{+})^2) q)\nonumber\\
  &+\max_{u^{-}\geq 0}(-\lambda (u^{-})^2 -\varepsilon u^{-} + (-u^{-} +\varepsilon u^{-} y + \lambda y (u^{-})^2) q).\label{optimaluhjb}
\end{align}
Suppose the ``second-order condition'' $q(y) y <1$ is satisfied (this holds for the function $q$ constructed in Lemma \ref{odesolution}). Then, the maxima are attained at
\begin{align}
u^{+}(y) &=\max\left(\frac{1}{2\lambda}\left(\frac{q(y)}{(1-y q(y))}- \varepsilon\right),0\right),\label{optbuyingrate}\\
u^{-}(y) &= \max\left(-\frac{1}{2\lambda}\left(\frac{q(y)}{(1-y q(y))}+ \varepsilon\right),0\right).\label{optsellingrate}
\end{align}
The optimizer with proportional transaction costs but without price impact is characterized by a no-trade interval around the frictionless optimum $y_\ast= \mu/\gamma \sigma^2$. Hence, we conjecture that the no-trade region in the present setting, 
\begin{equation*}
\Bigl\{y: u^+ (y) = u^-(y) =0 \Bigr\}=\Bigl\{y: -\varepsilon< \frac{q(y)}{1-y q(y)} < \varepsilon \Bigr\},
\end{equation*}
is also given by some interval $[y_-,y_+]$. Substituting the optimal turnover rates~(\ref{optbuyingrate}-\ref{optsellingrate}) back into~\eqref{optimaluhjb}, the HJB equation in turn simplifies to the ODE \eqref{eq:odeq}. Imposing continuity across the boundaries $y_-, y_+$ of the no-trade region in turn yields (\ref{abelend1}-\ref{abelend}). Since the differential equation~\eqref{eq:odeq} is of order one and there are four unknowns to be determined ($\beta, y_-, y_+$, and $q$), the value matching conditions~(\ref{abelend1}-\ref{abelend}) are not sufficient to characterize the solution. As a way out, we add two additional boundary conditions that become active when the investor's portfolio approaches full safe ($Y_t=0$) or full risky investment ($Y_t=1$). The idea is that the trading rate (\ref{optbuyingrate}-\ref{optsellingrate}) should remain finite in each case; moreover, it should be positive at $Y_t=0$ and negative at $Y_t=1$ so as to keep the risky weight in $[0,1]$.\footnote{Recall that this is needed 
to ensure solvency, since portfolios involving short or levered positions lead to bankruptcy with positive probability.} Solving the ODE \eqref{eq:odeq} at $y\in \{0,1\}$ leads to a quadratic equation for the boundary value of $q$; choosing the solution with the correct sign in turn gives
\begin{align}
q(0^+)&= \varepsilon + 2 \sqrt{\lambda \beta},\label{initial condition}\\
q(1^-)&= \frac{\lambda d -\varepsilon (1-\varepsilon)- \sqrt{\lambda d (\lambda d -2 + 2 \varepsilon)}}{(1-\varepsilon)^2}, \quad d = -\gamma \sigma^2- 2 \beta+ 2 \mu.\label{terminal condition}
\end{align}
 Together with the value matching conditions (\ref{abelend1}-\ref{abelend}), this yields the representation from Theorem~\ref{mainresult}. For $y_\ast \in (0,1)$, this informal derivation indeed leads to the correct answer (cf.\ the rigorous verification theorem in Section \ref{sec:proofs}). For $y_\ast \notin (0,1)$, however, the candidate risky weight explodes with positive probability, and the massive rebalancing that comes along with this reduces the corresponding wealth to zero (cf.\ Lemma \ref{ybound}). Hence, our candidate strategy is not even admissible in this case, and a simple buy-and-hold strategy turns out to be optimal instead (cf.\ Proposition \ref{prop:hold}). This stresses the need for rigorous verification theorems to complement heuristic considerations, which might otherwise lead to wrong results.

\section{Asymptotic Results}\label{sec:asymptotics}

The differential equation~\eqref{eq:odeq} in Theorem \ref{mainresult} is of Abel type; no explicit solution is known. However, it is possible to obtain asymptotic expansions for the no-trade region $[y_-,y_+]$, the trading rate $\hat{u}$, and the corresponding equivalent safe rate $\beta$ as the market frictions tend to zero. Then, the computation of the optimal policy and welfare is simplified from the solution of a nonlinear free-boundary value problem to finding the root of a nonlinear function, similarly as in Korn \cite{korn.98} for proportional and fixed costs. 

Consider the limiting regime where both the proportional transaction cost $\varepsilon$ \emph{and} the price impact parameter $\lambda$ tend to zero. If these frictions are considered separately, their leading-order impact on the equivalent safe rate is of order $\varepsilon^{2/3}$ and $\lambda^{1/2}$, respectively (cf.~\cite[Formula (2.7)]{gerhold.al.14} resp.~\cite[Formula (2.13)]{guasoni.weber.13}). To obtain an expansion in which neither friction vanishes, we rescale them appropriately to put their asymptotic contributions on the same scale:
\begin{equation}\label{rescaling}
\lambda=K\varepsilon^{4/3}, \quad \mbox{where } K>0.
\end{equation}
With this rescaling, the joint impact of the two frictions on the equivalent safe rate is of order $\varepsilon^{2/3}$ as well:
\begin{myprop}\label{limit_esr}
Assume $\lambda=K\varepsilon^{4/3}$ and define $c(\varepsilon):=\frac{\mu^2}{2\gamma\sigma^2}-\esr_\gamma(\hat u)$, where $\hat u$ is the long-run optimal strategy defined in Theorem~\ref{mainresult}. Then
\begin{equation*}
0<\liminf_{\varepsilon\rightarrow0}\frac{c(\varepsilon)}{\varepsilon^{2/3}}\leq\limsup_{\varepsilon\rightarrow0}\frac{c(\varepsilon)}{\varepsilon^{2/3}}<+\infty.
\end{equation*}
\end{myprop}

\begin{proof}
See Section~\ref{proof:prop_asymptotics}.
\end{proof}

The problem of finding the asymptotic expansion splits into two regions: close to the no-trade region\footnote{For small transaction costs -- without market price impact -- the no-trade region $[y_-,y_+]$ contains the Merton proportion $y_\ast$, its width is of order $\varepsilon^{1/3}$, and the corresponding welfare effect is of order $\varepsilon^{2/3}$, see~\cite{janecek.shreve.04}. With price impact, the no-trade region is again an interval, which contains the Merton proportion provided that $\varepsilon$ is not too small compared to $\lambda$ (cf.~ the discussion after Lemma~\ref{odesolution}).} and far away from it. To study the first regime, one can argue using homogenization techniques as in Soner and Touzi~\cite{soner.touzi.13} (see also~\cite{possamai.al.13,altarovici.al.13,moreau.al.14}). To this end, first derive the HJB equation for the value function $V$ and postulate an appropriate asymptotic expansion. Then, substitute the expansion back into HJB equation, and collect the leading-order terms. This in turn leads to the so-called \emph{corrector} equation. In the pure transaction cost case, the appropriate ansatz for $V$ is
\begin{equation}\label{homoansatz}
V(t,x,y) = V_0(t,x)-\varepsilon^{2/3} u(t,x) -\varepsilon^{4/3} w(t,x,z)+ \mathcal{O}(\varepsilon),
\end{equation}
with the fast variable $z= \frac{y-y_*}{\varepsilon^{1/3}}.$ Here, $V_0$ denotes the frictionless value function, $x$ the total wealth, and $y$ the wealth invested in the risky asset. Note that the HJB equation for $V$ is of second order. Hence, taking the second derivative of the function $w$ and multiplying it with $\varepsilon^{4/3}$ produces again a term which is of order $\varepsilon^{2/3}$ justifying the power $4/3$ in the expansion.

Our long-run objective forces us to work with the reduced value function $q$, i.e.,
\begin{equation}\label{valuefunctionheuristic}
V(t,x,y) = \frac{x^{1-\gamma}}{1-\gamma}e^{(1-\gamma)(\beta (T-t)+\int_{p}^y q(w) dw)}, \quad  p\in \mathbb{R}.
\end{equation} 
Whence, the homogenization approach has to be adapted as follows. Define
\begin{align}\label{transformationzqb}
z:=\frac{y-y_\ast}{\varepsilon^{1/3}}, \quad q(y):= \varepsilon \times
\begin{cases}
 r_B(z), &y \in [0,y_-],\\
 r(z), &y \in [y_-,y_+],\\
 r_S(z),  &y \in [y_+,1],
\end{cases}
\quad
\beta:=\frac{\mu^2}{2\gamma\sigma^2}-\varepsilon^{2/3}l,
\end{align}
for constants  $0 \leq y_- \leq y_+ \leq 1$, $l>0$, and functions $r_B,r,r_S$ to be determined.  

To find the boundaries of the no-trade region $y_-$ and $y_+$ we need two more conditions. To this end, we consider the limit behavior of $q$ as the risky weight $y$ approaches the full investment levels $0$ and $1$. Matching the terms at the leading order $\varepsilon^{2/3}$ we define $q(y)=\varepsilon^{2/3}q^*(y)$. As $\varepsilon \downarrow 0$, the ODE~\eqref{eq:odeq} then simplifies to
\begin{equation}\label{growthconditionr}
q^*(y)=\left(2K\gamma\sigma^2 \right)^{1/2}(y_\ast-y).
\end{equation}
This equation describes the behavior of $q^*$ (and in turn $q$) away from the no trade region and therefore yields the required conditions to determine $y_-$ and $y_+$.\footnote{A similar ``pasting'' of two different expansions is used in \cite{chen.dai.13} to deal with a small capital gains tax.}

As $\varepsilon\downarrow 0$, the transformation~\eqref{transformationzqb} reduces the ODE~\eqref{eq:odeq} to an inhomogeneous Riccati equation:
\begin{align}
0 &=
-\frac{\gamma\sigma^2}{2}z^2+l
+\frac{\sigma^2}{2}y_\ast^2(1-y_\ast)^2r'_B+\frac{1}{4K}(r_B-1)^2,\quad  z \in(-\infty,z_-],\label{leftr}\\
0 &= 
-\frac{\gamma\sigma^2}{2}z^2+l
+\frac{\sigma^2}{2}y_\ast^2(1-y_\ast)^2r', \qquad  \qquad \qquad \quad z \in[z_-,z_+],\label{midr}\\
0 &= 
-\frac{\gamma\sigma^2}{2}z^2+l
+\frac{\sigma^2}{2}y_\ast^2(1-y_\ast)^2r'_S+\frac{1}{4K}(r_S+1)^2,\quad  z \in[z_+,+\infty),\label{rightr}
\end{align}
where the rescaled buying and selling boundaries $z_-$ and $z_+$ are determined by the following value matching conditions:
\begin{align}\label{valuematchat1}
r_B(z_-) =r(z_-)= 1, \quad r_S(z_+) =r(z_+) = -1.
\end{align}
As $z \rightarrow -\infty$ (resp. $z \rightarrow \infty$), the function $r_B$ (resp. $r_S$) must diverge with the same rate as in our first expansion in \eqref{growthconditionr}:
\begin{equation}\label{rlminusinfty}
\lim_{z \rightarrow -\infty} \frac{r_B(z)}{-\sqrt{2 K \gamma \sigma^2} z}=1 \quad \text{and} \quad \lim_{z \rightarrow \infty} \frac{r_S(z)}{-\sqrt{2 K \gamma \sigma^2} z}=1.
\end{equation}
In summary, this leads to the ODE~(\ref{leftr}-\ref{rightr}) with value matching and growth conditions~(\ref{valuematchat1}-\ref{rlminusinfty}) -- the corrector equation in the present setting.

Inserting~\eqref{growthconditionr} into~\eqref{valuefunctionheuristic} with the expansion for the equivalent safe rate in~\eqref{transformationzqb} yields that:
\begin{align}
V(t,x,y) &= V_0(t,x)\cdot e^{(1-\gamma)(-\varepsilon^{2/3}l\cdot(T-t)+ \varepsilon^{2/3} \int_{p}^y q^*(w) dw)}\nonumber\\
 &= V_0(t,x) -V_0(t,x)(1-\gamma)(T-t)l \cdot \varepsilon^{2/3}\nonumber\\
 & \quad + V_0(t,x)(1-\gamma) \int_{p}^y q^*(w) dw \cdot \varepsilon^{2/3} + o(\varepsilon^{2/3}). \label{valuefunctionansatz}
\end{align} 

The next proposition (proved in Section~\ref{proof:prop_asymptotics}) shows that the ODE~(\ref{leftr}-\ref{rightr}) with boundary conditions~(\ref{valuematchat1}-\ref{rlminusinfty}) admits a unique solution. As a side product, we show that additional price impact decreases the width of the no-trade region compared to the pure linear-cost case.

\begin{myprop}\label{prop_asymptotics}
There are unique $z_-$, $z_+$ and $l_*>0$ such that equation (\ref{leftr}-\ref{rightr}) has a solution that satisfies the value matching conditions~\eqref{valuematchat1} and the growth conditions \eqref{rlminusinfty}. In particular,
\begin{equation*}
l_*>\max\left(\sqrt{\frac{\gamma K}{2}}\sigma^3y_\ast^2(1-y_\ast)^2,\left(\frac{3}{4}\sqrt{\frac{\gamma}{2}}\sigma^3y_\ast^2(1-y_\ast)^2\right)^{2/3}\right).
\end{equation*}
Also, if $z^0_-$ and $z^0_+$ are the boundaries of no-trade region without price impact (see ~\cite[Formula (2.9)]{gerhold.al.14}), then $z_->z^0_-$ and $z_+<z^0_+$.
\end{myprop}

\begin{proof}
See Section~\ref{proof:prop_asymptotics}.
\end{proof}

We have shown with formal arguments that the appropriately rescaled value function $q(y)$ in Theorem~\ref{mainresult} converges to the solution $r(z)$ of equation (\ref{leftr}-\ref{rightr}) with growth conditions \eqref{rlminusinfty}. The next proposition makes these heuristic arguments rigorous. 

\begin{myprop}\label{prop_convergence}
Assume $\lambda=K\varepsilon^{4/3}$. Let $q^\varepsilon(y)$ be the solution to equation~\eqref{eq:odeq} as in Theorem~\ref{mainresult}. Then $\varepsilon^{-1}q^\varepsilon(\varepsilon^{1/3}z+y_*)$ converges to the solution of equation (\ref{leftr}-\ref{rightr}) defined in Proposition~\ref{prop_asymptotics}. In particular, $\lim_{\varepsilon\rightarrow0}\frac{c(\varepsilon)}{\varepsilon^{2/3}}=l_*$, where $c(\varepsilon):=\frac{\mu^2}{2\gamma\sigma^2}-\esr_\gamma(\hat u)$.
\end{myprop}

\begin{proof}
See Section~\ref{proof:prop_asymptotics}.
\end{proof}

To find the explicit solution of the ODE~(\ref{leftr}-\ref{rightr}) with boundary conditions~(\ref{valuematchat1}-\ref{rlminusinfty}) we proceed as follows. A simple calculation shows that the linear, inhomogeneous ODE~\eqref{midr} has the explicit solution
\begin{equation*}
r (z) = \frac{2}{\sigma^2 y_\ast^2 (1-y_\ast)^2}\left(\frac{\gamma \sigma^2}{6} z^3 - l z\right).
\end{equation*}  
Since $r$ is an odd function, the boundary conditions $r(z_-)=1$ and $r(z_+)=-1$ can be replaced by $r(z_-)=1$ and $r(0)=0$. Then, $z_-=-z_+$ and it suffices to determine $z_-$ because the constant $l$ is linked to $z_-$ via the condition $r(z_-)=1$:
\begin{equation}\label{eq:l}
l(z_-) = \frac{\gamma \sigma^2}{6} z_-^2-\frac{\sigma^2 y_\ast^2(1-y_\ast)^2}{2 z_-}.
\end{equation}  
To find the rescaled trading boundary $z_-$ determined by $r_B(z_-)=1$, the Riccati equation~\eqref{leftr} with boundary condition~\eqref{rlminusinfty} needs to be solved first. 

 Equation~\eqref{leftr} is -- up to transformations~\cite[Transformations 1.25, 1.105, 2.220, and 2.273, Formula (9)]{kamke.77} -- equivalent to a Whittaker equation for which explicit solutions (in terms of the Whittaker functions) are known (cf.~\cite[Section 1]{slater.60}). For brevity, we omit the derivation of the solution and simply state the final result:
\begin{equation}\label{solutionrl}
r_B (z,l) := -\frac{1}{z}\left(\frac{1}{2a}+ \frac{c}{2 a} \sqrt{\frac{1}{2 K \gamma \sigma^2}}\right)+1 +\sqrt{2 K \gamma \sigma^2} z-\frac{2}{a z} \frac{W\left(k+1,-1/4,a \sqrt{2K \gamma \sigma^2} z^2\right)}{W\left(k,-1/4,a \sqrt{2K \gamma \sigma^2} z^2\right)},
\end{equation}
where
\begin{align*}
a = \frac{1}{2 K \sigma^2 y_\ast^2 (1-y_\ast)^2}, \quad c = \frac{2 l}{\sigma^2 y_\ast^2 (1-y_\ast)^2}, \quad
k = \frac{c}{4} \sqrt{\frac{1}{2 K \gamma \sigma^2}},
\end{align*}
and the Whittaker function $W$ is defined via the Kummer function\footnote{\label{kummer}For $b\not= 0,-1,-2,\cdots,$ the Kummer function$\,_1F_1 (a,b,x)$ is defined through the following absolute convergent series \cite[Chapter 1]{slater.60}:
\begin{align*}
\,_1F_1 (a,b,x):= \sum_{n=0}^{\infty} \frac{a^{(n)}x^n}{b^{(n)}n!}.
\end{align*} 
Here, the Pochhammer symbol $a^{(n)}$ is given by $a^{(n)}:= a(a+1)(a+2) \cdots (a+n-1)$. For $a\not= 0,-1,-2,\cdots,$ we can write $a^{(n)}= \Gamma{(a+n)}/\Gamma{(a)}$, where $\Gamma(x)$ denotes the Gamma function.} $\,_1F_1 (a,b,x)$ and the Tricomi function $U(\xi,\eta,x)$ cf.~\cite[Chapter 1]{slater.60}:
\begin{align*}
U(\xi,\eta,x):= &\frac{\Gamma(1-\eta)}{\Gamma(1+\xi-\eta)} \,_1F_1(\xi,\eta,x)+\frac{\Gamma(\eta-1)}{\Gamma(\xi)} x^{1-\eta} \,_1F_1(1+\xi-\eta,2-\eta,x),\\
W(k,m,x) :=& x^{\frac{1}{2}+m}e^{-\frac{1}{2}x} U\left(1/2+m-k,1+2m,x\right).
\end{align*}
We note that the general solution of the Whittaker equation is given by a linear combination of the Whittaker functions\footnote{Note that the Whittaker function $W$ (cf.~\cite[Formula 1.7.1]{slater.60}) can be written as 
\begin{align*}
W(k,m,x) &= \frac{\pi}{\sin(2 m \pi)}\left(\frac{-M(k,m,x)}{\Gamma(1/2-m-k)\Gamma(1+2m)}+\frac{M(k,-m,x)}{\Gamma(1/2+m-k)\Gamma(1-2m)}\right).
\end{align*}
} $M(k,m,x)$ and $M(k,-m,x)$ (cf.~\cite[Section 16]{whittaker.watson.96}) with 
\begin{align*}
M(k,m,x) &:= x^{\frac{1}{2}+m}e^{-\frac{1}{2}x} \,_1F_1 \left(1/2+m-k,1+2m,x\right).
\end{align*}
However, the Riccati equation needs to be solved with the initial condition~\eqref{rlminusinfty}. Therefore, the Whittaker function $W(k,m,x)$ is the only candidate, since it has the correct asymptotic growth~\cite[Section 16.31]{whittaker.watson.96}:
\begin{equation}\label{whittakerasymptotic}
W(k,m,x) \sim x^{k}e^{-\frac{1}{2}x}, \quad \text{as} \quad x\rightarrow \infty.
\end{equation}
Indeed, for any $l\in \mathbb{R}$, \eqref{solutionrl} satisfies the boundary condition~\eqref{rlminusinfty}:
\begin{align*}
\lim_{z \rightarrow -\infty} \frac{r_B(z,l)}{-\sqrt{2 K \gamma \sigma^2} z} &= \lim_{z \rightarrow -\infty} \frac{1}{-\sqrt{2 K \gamma \sigma^2} z} \left(-\frac{1}{z}\left(\frac{1}{2a}+ \frac{c}{2 a} \sqrt{\frac{1}{2 K \gamma \sigma^2}}\right)+1\right)\\
 &  + \frac{\sqrt{2 K \gamma \sigma^2}}{-\sqrt{2 K \gamma \sigma^2} z}  \left(z-\frac{2 z}{a \sqrt{2 K \gamma \sigma^2} z^2} \frac{W\left(k+1,-1/4,a \sqrt{2K \gamma \sigma^2} z^2\right)}{W\left(k,-1/4,a \sqrt{2K \gamma \sigma^2} z^2\right)}\right)\\
  &= 1,
\end{align*}
where we have used~\eqref{whittakerasymptotic} in the last step.

Now, we can put everything together. $z_-$ is determined by $r_B(z_-,l(z_-))=1$ with $l(z_-)$ from \eqref{eq:l}. Hence,  the asymptotic expansions for the growth rate $\beta$ and the boundaries $y_-,y_+$ of the no-trade region in \eqref{transformationzqb} are given by:
\begin{align*}
\beta &= \frac{\mu^2}{2 \gamma \sigma^2}- \left(\frac{\gamma \sigma^2}{6} z_-^2-\frac{\sigma^2 y_\ast^2 (1-y_\ast)^2}{2 z_-}\right) \varepsilon^{\frac{2}{3}}+ \mathcal{O}(\varepsilon),\\
y_{\mp} &= y_\ast \pm z_- \varepsilon^{\frac{1}{3}}+ \mathcal{O}(\varepsilon^{\frac{2}{3}}),
\end{align*}
where $z_-$ is the root of the equation $r_B(z,l(z)) =1$ with $l(z)$ from \eqref{eq:l}. The asymptotics for the turnover rate $\hat{u}$ can be derived similarly. Recall from Theorem~\ref{mainresult} that
\begin{equation*}
\hat{u}(y)=
\begin{cases}
\frac{1}{2 \lambda}\left(\frac{q(y)}{1-y q(y)}-\varepsilon\right) \geq 0, &\text{if } y \in [0,y_-],\\
0, &\text{if } y \in [y_-,y_+],\\
\frac{1}{2 \lambda}\left(\frac{q(y)}{1-y q(y)}+\varepsilon\right) \leq 0, &\text{if } y \in [y_+,1].\\
\end{cases}
\end{equation*}  
Using the same transformations as in~\eqref{transformationzqb} for $y$ and $q$, a straightforward calculation shows that 
\begin{equation}\label{eq:hatu}
\hat{u}(y) = 
\begin{cases}
\frac{1}{2K} \left(r_B(z)-1\right)\varepsilon^{-\frac{1}{3}}+ \mathcal{O}(1), &\text{if } z = (y-y_\ast) \varepsilon^{-1/3} \in (-\infty,z_-],\\
0, &\text{if } z = (y-y_\ast) \varepsilon^{-1/3} \in [z_-,z_+],\\
\frac{1}{2K} \left(r_S(z)+1\right)\varepsilon^{-\frac{1}{3}}+ \mathcal{O}(1), &\text{if } z = (y-y_\ast) \varepsilon^{-1/3} \in [z_+,+\infty),\\
\end{cases}
\end{equation} 
where we abbreviate $r_B (z) := r_B (z, l(z_-))$ and $r_S (z) :=r_B(z)-2$. 

In summary, asymptotically for small trading costs, the solution of the nonlinear free-boundary problem \eqref{eq:odeq} can be reduced to finding the root of a scalar function. The approximation~\eqref{eq:hatu} performs very well, even for relatively large values of the asymptotic parameters $\varepsilon, \lambda$, cf. Figure~\ref{fig:asymptotics}. Moreover, it also allows to say more about the structure of the optimal turnover rate near the trading boundaries and far away from these. 

Far away from the no-trade region, i.e., as $z\rightarrow -\infty$ or, equivalently,  $y\downarrow 0$ (resp. $z\rightarrow \infty$, i.e., $y \uparrow 1$) the boundary conditions
\begin{equation*}
\lim_{z\rightarrow -\infty} \frac{r_B(z)}{-\sqrt{2 K \gamma \sigma^2} z}=1, \quad \text{resp.} \quad \lim_{z\rightarrow \infty} \frac{r_S(z)}{-\sqrt{2 K \gamma \sigma^2} z}=1,
\end{equation*} 
imply
\begin{align}
\hat{u} (y) &= \frac{1}{2K} \left(-\sqrt{2 K \gamma \sigma^2} z -1\right) \varepsilon^{-1/3} + \mathcal{O}(1) \notag\\
 &=
\frac{1}{2K}  \sqrt{2 K \gamma \sigma^2} (y_\ast-y) \varepsilon^{-2/3}+ \mathcal{O}(\varepsilon^{-1/3}) \notag\\
 &= \sigma \sqrt{\frac{\gamma}{2}} (y_\ast-y)\lambda^{-1/2}+ \mathcal{O}(\lambda^{-1/4}).\label{eq:asymp2}
\end{align}
Hence, for large deviations from the trading boundaries we recover the leading-order expansion of the wealth turnover without proportional transaction costs~\cite[Formula (2.12)]{guasoni.weber.13}.

Close to the trading boundaries, we can apply Taylor's Theorem to get a first-order approximation as well. To this end, we first compute the derivative of the corresponding Whittaker functions. The differential property~\cite[Formula (2.4.24)]{slater.60} and the recurrence relation~\cite[Formula (2.5.11)]{slater.60} of $W(k,m,x)$ show that
\begin{align*}
\left(\frac{W(k+1,-1/4,x)}{x W(k,-1/4,x)}\right)' =& \frac{W(k+1,-1/4,x)}{x W(k,-1/4,x)}\left(\frac{W'(k+1,-1/4,x)}{W(k+1,-1/4)}-\frac{W'(k,-1/4,x)}{W(k,-1/4,x)}-\frac{1}{x}\right) \\
 =& \frac{W(k+1,-1/4,x)}{x W(k,-1/4,x)} \Biggl[ \frac{W(k+1,-1/4,x)}{x W(k,-1/4,x)} -\frac{2}{x} \\
 & \quad -\frac{1}{x^2}\left(\left(-\frac{1}{4}-(k+1)+\frac{1}{2}\right)\left(-\frac{1}{4}+(k+1)-\frac{1}{2}\right)\frac{x W(k,-1/4,x)}{W(k+1,-1/4,x)}\right)\\
 & \quad -(2(k+1)x-x^2)\Biggr],
\end{align*}
where $x= a \sqrt{2 K \gamma \sigma^2} z^2$. Taking into account the value matching condition $r_B(z_-)=1$, a straightforward computation yields that the wealth turnover close to the trading boundaries $y_\pm$ is given by
\begin{equation}\label{eq:hatuapproxnearzm}
\hat{u}(y) = 
\begin{cases}
\frac{\varepsilon^{-1/3}}{2K} \times F \times (z-z_-) + \mathcal{O}(1), &\text{if } z = (y-y_\ast) \varepsilon^{-1/3} < z_-,\\
\frac{\varepsilon^{-1/3}}{2K} \times F \times (z-z_+) + \mathcal{O}(1), &\text{if } z = (y-y_\ast) \varepsilon^{-1/3} > z_+,
\end{cases}
\end{equation} 
where
\begin{align*}
x_- &:= a \sqrt{2 K \gamma \sigma^2} z_-^2,\\
D &:= \frac{1}{2}\left(1- \frac{1}{\sqrt{2 K \gamma \sigma^2} z_-^2}\left(\frac{1}{2a}+\frac{c}{2a}\sqrt{\frac{1}{2 K \gamma \sigma^2}}\right)\right),\\
E &:= D\left(D-\frac{2}{x_-}-\frac{1}{x_-^2}\left(\frac{1}{D}\left(-\frac{1}{4}-(k+1)+\frac{1}{2}\right)\left(-\frac{1}{4}+(k+1)-\frac{1}{2}\right)-(2(k+1)x_- -x_-^2)\right)\right),\\
F &:= \left(\frac{1}{2a}+\frac{c}{2a}\sqrt{\frac{1}{2 K \gamma \sigma^2}}\right)\frac{1}{z_-^2}+\sqrt{2 K \gamma \sigma^2}-2 \sqrt{2 K \gamma \sigma^2} (D+2 a \sqrt{2 K \gamma \sigma^2} E z_-^2).
\end{align*}
\begin{figure}
\subfigure{\includegraphics[width=0.49\textwidth]{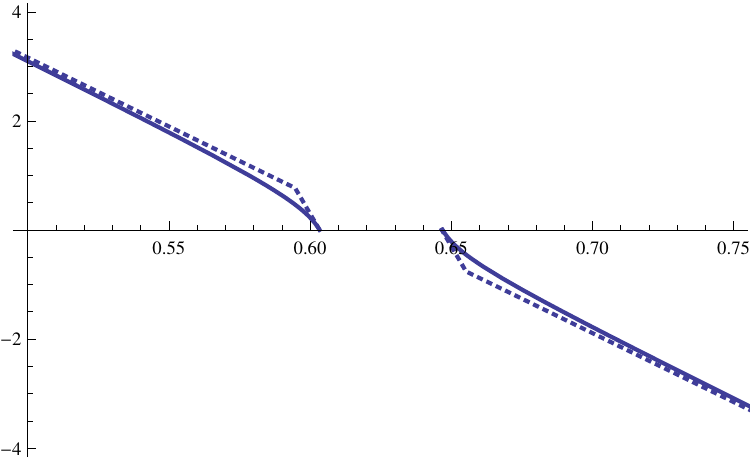}}\hfill
\subfigure{\includegraphics[width=0.49\textwidth]{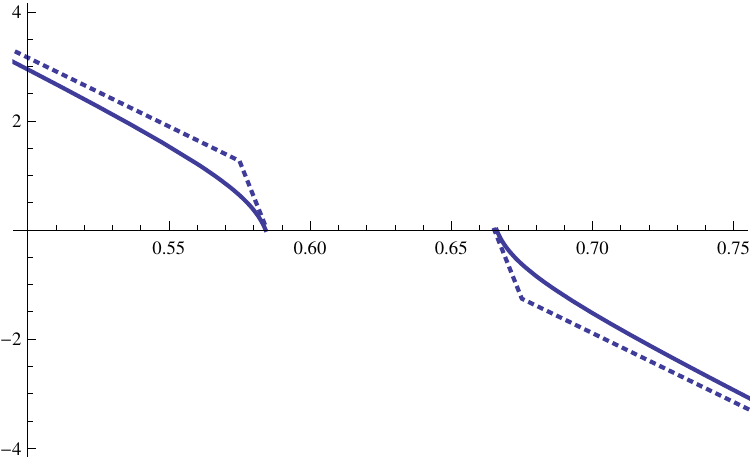}}
\caption{optimal wealth turnover $\hat{u}$ from Theorem~\ref{mainresult} (solid), and its piecewise linear approximation (dotted) by~\eqref{eq:hatuapproxnearzm} close to the trading boundaries and~\eqref{eq:asymp2} further away from these. Model parameters are $\mu = 8\%$, $\sigma = 16\%$, $\gamma = 5$, $\lambda= 0.01\%$, and $\varepsilon=0.1\%$ (left panel) resp.\ $\varepsilon=0.5\%$ (right panel).}
\label{fig:turnovernearzm}
\end{figure}

In particular, for small deviations from the trading boundaries the wealth turnover -- at the first order -- is again linear, however with a different slope. This is illustrated in Figure \ref{fig:turnovernearzm}. If the proportional costs are small compared to the price impact, then the two slopes are very similar. For larger spreads, however, turnover grows significantly faster near the trading boundaries, compare the right panel in Figure \ref{fig:turnovernearzm}.

\paragraph{More General Asymptotics}\label{ss:robust}
In this section we argue informally that -- in the small-cost limit -- the structure of the solution obtained here for the specific price-impact model~\eqref{modelassump} remains valid rather generally. This is consistent with results for linear and quadratic costs treated separately \cite{kallsen.muhlekarbe.13,moreau.al.14}\footnote{Making these arguments rigorous would require to specify suitable integrability and boundary conditions for the additional state variables that appear in this context.} 

To this end, consider an exogenously given diffusion process $\xi$ with dynamics
\begin{equation}\label{eq:xi}
d \xi_t= \mu_\xi(\xi_t) dt + \sigma_\xi (\xi_t) dW_t,
\end{equation}
where $\mu_\xi, \sigma_\xi$ are smooth functions such that the SDE~\eqref{eq:xi} is well-defined. This additional state variable can model the wealth of other agents trading in the market, for example. The average execution price for the large investor is in turn given by
$$S_t \left(1+ \varepsilon \cdot \mathrm{sgn}(\Delta \theta)+ \lambda(\xi_t,X_t) \frac{S_t \Delta \theta}{\Delta t}\right).$$
To wit, the price impact is a general function of the large investor's wealth $X$ and the exogenous process $\xi$. If $\lambda(\xi,x)= \lambda /x$ we recover the model~\eqref{modelassump}. 

An easy application of It\^o's Lemma shows that in this setting, the wealth process $X$ and the risky weight $Y$ satisfy:
\begin{align}
\frac{d X_t}{X_t} &= Y_t (\mu dt + \sigma dW_t)- \varepsilon |u_t| dt -\lambda(\xi_t,X_t) X_t u_t^2 dt,\label{dynamicsXnew}\\
dY_t &= \left(Y_t(1-Y_t) (\mu-Y_t \sigma^2)+u_t+ \varepsilon |u_t| Y_t + \lambda(\xi_t,X_t) X_t Y_t u_t^2 \right)dt+ Y_t(1-Y_t) \sigma dW_t.\label{dynamicsYnew}
\end{align}

Let $V(t,X_t,Y_t,\xi_t)$ denote the corresponding finite horizon value function, which is assumed to depend on the current wealth $X_t$, the current risky weight $Y_t$, the exogenous state $\xi_t$, and time $t$.  Arguing as in Section~\ref{heuristic}, it follows that 
\begin{equation}\label{reducedvaluefunctionnew}
V(t,x,y,\xi) = \frac{x^{1-\gamma}}{1-\gamma}e^{(1-\gamma)(\beta (T-t)+\int_{y_*}^y q(\xi,x,u) du)},
\end{equation}
where the function $q(\xi,x,y)$ satisfies
\begin{align}
0= &-\beta+ \mu y -\frac{\gamma \sigma^2}{2} y^2+ y(1-y) (\mu-\gamma \sigma^2 y) q + \frac{\sigma^2}{2}y^2(1-y)^2(q_y+(1-\gamma)q^2)\nonumber\\
 & + \mu y x \int_{y_*}^y q_x(\xi,x,u) du + \sigma^2 y^2(1-y) \left((1-\gamma)xq(\xi,x,y) \int_{y_*}^y q_x(\xi,x,u) du+ x q_x(\xi,x,y)\right)\nonumber\\
 &+ \frac{\sigma^2 y^2}{2}\Biggl[2(1-\gamma) x \int_{y_*}^y q_x(\xi,x,u) du + (1-\gamma)\left(x\int_{y_*}^y q_x(\xi,x,u) du\right)^2\nonumber\\
 &+ x^2\int_{y_*}^y q_{xx}(\xi,x,u) du\Biggr]\nonumber\\
 & + \mu_\xi(\xi) \int_{y_*}^y q_\xi(\xi,x,u) du+ \frac{\sigma_\xi^2(\xi)}{2} \left(\left(\int_{y_*}^y q_\xi(\xi,x,u) du\right)^2 (1-\gamma) + \int_{y_*}^y q_{\xi\xi}(\xi,x,u) du\right)\nonumber\\
 & +\sigma \sigma_\xi (\xi) y (1-y)\left((1-\gamma) q(\xi,x,y) \int_{y_*}^y q_\xi(\xi,x,u) du + q_\xi(\xi,x,y)\right)\nonumber\\
 & +\sigma \sigma_\xi (\xi) y \Biggl[(1-\gamma) \int_{y_*}^y q_\xi(\xi,x,u) du + (1-\gamma) x \int_{y_*}^y q_x(\xi,x,u)du \int_{y_*}^y q_{\xi}(\xi,x,u)du\nonumber\\
 &+ x \int_{y_*}^y q_{\xi x}(\xi,x,u)du\Biggr]\nonumber\\
&+\begin{cases}\label{eq:odeqnew}
\frac{1}{4 \lambda(\xi,x)x}\frac{(q-\varepsilon(1-y q+x \int_{y_*}^y q_x(\xi,x,u)du))^2}{1-y q+ x \int_{y_*}^y q_x(\xi,x,u)du}, &\text{if } y\in [0,y_-(\xi,x)],\\
 0, &\text{if } y\in[y_-(\xi,x),y_+(\xi,x)],\\
\frac{1}{4 \lambda(\xi,x)x}\frac{(q+\varepsilon(1-y q+x \int_{y_*}^y q_x(\xi,x,u)du))^2}{1-y q+x \int_{y_*}^y q_x(\xi,x,u)du},  &\text{if } y\in [y_+(\xi,x),1].
\end{cases}
\end{align}
and
\begin{align}
 q(\xi,x,y_-(\xi,x))= &\frac{\varepsilon+\varepsilon x\int_{y_*}^{y_-(\xi,x)} q_x(\xi,x,u)du }{1+\varepsilon y_-(\xi,x)},\nonumber\\
 q(\xi,x,y_+(\xi,x))= &\frac{-\varepsilon-\varepsilon x\int_{y_*}^{y_+(\xi,x)} q_x(\xi,x,u)du}{1-\varepsilon y_+(\xi,x)},\nonumber
\end{align}
for functions $0\leq y_-(\xi,x) \leq y_+(\xi,x) \leq 1$ to be determined. To derive the corresponding small-cost asymptotics, rescale transaction costs and price impact accordingly:
$$\lambda(\xi,x)x = K(\xi,x) \varepsilon^{4/3},$$
for some function $K$. Close to the no-trade region we again apply the homogenization approach and define 
\begin{align}\label{transformationzqbnew}
z:=\frac{y-y_\ast}{\varepsilon^{1/3}}, \quad q(\xi,x,y):= \varepsilon \times
\begin{cases}
 r_B(\xi,x,z), &y \in [0,y_-(\xi,x)],\\
 r(\xi,x,z), &y \in [y_-(\xi,x),y_+(\xi,x)],\\
 r_S(\xi,x,z),  &y \in [y_+(\xi,x),1],
\end{cases}
\quad
\beta:=\frac{\mu^2}{2 \gamma \sigma^2}-l \varepsilon^{2/3},
\end{align}
for  a constant\footnote{Note that the optimal long-run growth rate is typically constant even in factor models with an additional state variable, compare \cite{guasoni.robertson.12}.}  $l>0$ and functions $0 \leq y_-(\xi,x) \leq y_+(\xi,x) \leq 1$, as well as $r_B(\xi,x,z)$, $r(\xi,x,z)$, $r_S(\xi,x,z)$ to be determined. 
To identify the boundaries of the no-trade region we need additional conditions. To this end, we again consider $q(\xi,x,y)= \varepsilon^{2/3} q^*(\xi,x,y)$. As $\varepsilon \downarrow 0$,  a direct calculation shows that the ODE~\eqref{eq:odeqnew} reduces to 
\begin{equation}\label{growthconditionrnew}
q^*(\xi,x,y) = (2 K (\xi,x) \gamma \sigma^2)^{1/2} (y_*-y).
\end{equation} 
If the functions $r_B,r,r_S$ are sufficiently smooth, then as $\varepsilon \downarrow 0$ transformation~\eqref{transformationzqbnew} reduces\footnote{For each $h\in\{r_B,r,r_S\}$, note that if $h_{\xi}(\xi,x,\cdot)$ is continuous then $$ \int_{y_*}^y q_{\xi} (\xi,x,u) du= \varepsilon \int_{y_*}^y h_{\xi} (\xi,x,(u-y_*)\varepsilon^{-1/3}) du = \varepsilon^{4/3} \int_{0}^z h_{\xi} (\xi,x,s) ds= \mbox{o}(\varepsilon^{2/3}).$$} the ODE~\eqref{eq:odeqnew} to 
\begin{align}
0 &=
-\frac{\gamma\sigma^2}{2}z^2+l
+\frac{\sigma^2}{2}y_\ast^2(1-y_\ast)^2r'_B+\frac{1}{4K(\xi,x)}(r_B-1)^2,\quad  z \in(-\infty,z_-(\xi,x)],\label{leftrnew}\\
0 &= 
-\frac{\gamma\sigma^2}{2}z^2+l
+\frac{\sigma^2}{2}y_\ast^2(1-y_\ast)^2r', \qquad  \qquad \qquad \quad z \in[z_-(\xi,x),z_+(\xi,x)],\label{midrnew}\\
0 &= 
-\frac{\gamma\sigma^2}{2}z^2+l
+\frac{\sigma^2}{2}y_\ast^2(1-y_\ast)^2r'_S+\frac{1}{4K(\xi,x)}(r_S+1)^2,\quad  z \in[z_+(\xi,x),+\infty),\label{rightrnew}
\end{align}
where the rescaled buying and selling boundaries $z_-(\xi,x)$ and $z_+(\xi,x)$ satisfy:
\begin{align}\label{valuematchat1new}
r_B(\xi,x,z_-(\xi,x)) =r(\xi,x,z_-(\xi,x))= 1, \quad r_S(\xi,x,z_+(\xi,x)) =r(\xi,x,z_+(\xi,x)) = -1.
\end{align}
As $z \rightarrow -\infty$ (resp. $z \rightarrow \infty$ ) the function $r_B(\xi,x,\cdot)$ (resp. $r_S(\xi,x,\cdot)$) must diverge with the same rate as in \eqref{growthconditionrnew}:
\begin{equation}\label{rlminusinftynew}
\lim_{z \rightarrow -\infty} \frac{r_B(\xi,x,z)}{-\sqrt{2 K(\xi,x) \gamma \sigma^2} z}=1 \quad \text{and} \quad \lim_{z \rightarrow \infty} \frac{r_S(\xi,x,z)}{-\sqrt{2 K(\xi,x) \gamma \sigma^2} z}=1.
\end{equation}
As a result, the asymptotic trading boundaries and trading rate are still determined by an inhomogenous Riccati equation, but by a different one for each value of the additional state variables. To wit, the asymptotic solutions for more general cost structures are obtained by plugging in the current value of the cost into the asymptotic expansions derived in the previous section.

\section{Proofs}\label{sec:proofs}
\subsection{Proof of Theorem~\ref{mainresult}}
Assume throughout that $y_\ast \in (0,1)$. The first step towards a rigorous verification theorem is to show that the differential equation~\eqref{eq:odeq} indeed admits a solution with the required properties~(\ref{abelini0}-\ref{abelend}). To this end, we first rewrite \eqref{eq:odeq} in slope field notation:\footnote{This will be justified a posteriori by the monotonicity properties of $q$ established in Lemma \ref{odesolution}.}
\begin{equation}
q'=f(y,q) = 
\begin{cases}
f_B(y,q), &q \geq \frac{\varepsilon}{1+\varepsilon y}, \\
f_{NT}(y,q), &-\frac{\varepsilon}{1-\varepsilon y}\leq q \leq \frac{\varepsilon}{1+\varepsilon y}, \\
f_S(y,q), &q \leq -\frac{\varepsilon}{1-\varepsilon y},
\end{cases}
\end{equation}
where
\begin{equation*}
\begin{split}
-\beta + \mu y -\frac{\gamma \sigma^2}{2} y^2+ y(1-y) (\mu-\gamma \sigma^2 y) q + \frac{\sigma^2}{2}y^2(1-y)^2(q'+(1-\gamma)q^2)\\
+\begin{cases}
\frac{1}{4 \lambda}\frac{(q-\varepsilon(1-y q))^2}{1-y q}&=:f_B(y,q),\\
 0&=:f_{NT}(y,q),\\
\frac{1}{4 \lambda}\frac{(q+\varepsilon(1-y q))^2}{1-y q}&=:f_S(y,q).
\end{cases}
\end{split}
\end{equation*}
Notice that $f(y,q)$ is well defined because $f_B(y,q)=f_{NT}(y,q)$ on $q =\frac{\varepsilon}{1+\varepsilon y}$ and $f_{NT}(y,q)=f_S(y,q)$
on $q = -\frac{\varepsilon}{1-\varepsilon y}$.

\begin{myrek}\label{beta}
Allocating the entire wealth into the riskless asset (resp. the risky asset) is an admissible strategy that does not require trading. The corresponding
equivalent safe rate is $0$ (resp. $\mu-\frac{\gamma \sigma^2}{2}$). This provides a natural lower bound for the optimal equivalent safe rate, namely
$\beta\geq\max\{\mu-\frac{\gamma \sigma^2}{2},0\}$. Conversely, an upper bound is given by the frictionless equivalent safe rate $\frac{\mu^2}{2\gamma\sigma^2}$.
\end{myrek}

\begin{mylemma}\label{odesolution}
Suppose $\lambda$ and $\varepsilon$ are sufficiently small. Then, for a suitable $\beta \in \big[\max\{\mu-\frac{\gamma \sigma^2}{2},0\}, \frac{\mu^2}{2\gamma \sigma^2} \big]$, there is a solution of
$q'=f(y,q)$ such that
\begin{align}
q(0^+) &= b_0(\varepsilon,\lambda):= \varepsilon + 2 \sqrt{\lambda \beta}, \label{in0}\\
q(1^-) &= b_1(\varepsilon,\lambda):= \frac{\lambda d - \varepsilon(1-\varepsilon)-\sqrt{\lambda d (\lambda d -2 + 2 \varepsilon)}}{(1-\varepsilon)^2}, \quad \mbox{with }d = -\gamma \sigma^2-2 \beta+ 2 \mu.\label{in1}
\end{align}
In particular, there exist $y_-, y_+ \in [0,1]$ satisfying~\eqref{eq:odeq} and
\begin{align}\label{decreasingq}
\begin{cases}
q(y) > \frac{\varepsilon}{1+\varepsilon y},\ \qquad \quad \quad \text{if } y \in [0,y_-),\\
\frac{-\varepsilon}{1-\varepsilon y} \leq q(y) \leq \frac{\varepsilon}{1+\varepsilon y},\ \quad \text{if } y \in [y_-,y_+],\\
q(y) \leq \frac{-\varepsilon}{1-\varepsilon y},\ \qquad \quad \quad  \text{if } y \in (y_+,1].
\end{cases}
\end{align}
Moreover, the solution $q$ fulfills $q(y) y < 1$ for all $y\in[0,1]$.
\end{mylemma}

\begin{proof}

First, notice that for every $y\in(0,1)$ we have $\frac{1}{y}>\frac{\varepsilon}{1+\varepsilon y}$, and thus 
$\lim_{q\rightarrow(\frac{1}{y})^-}f(y,q)=\lim_{q\rightarrow(\frac{1}{y})^-}f_B(y,q)=-\infty$. 
Hence, every solution to $q'=f(y,q)$ starting below the curve $\frac{1}{y}$ must remain below this curve. The rest of the proof proceeds as follows:
\begin{itemize}
\item [(i)]  For every $\beta>\max\{\mu-\frac{\gamma \sigma^2}{2},0\}$, there is a unique solution $q_0^\beta(y)$ defined on $(0,y_0)\subset(0,1)$ that satisfies the boundary condition (\ref{in0}) at $y=0$
and a unique solution $q_1^\beta(y)$ defined on $(y_1,1)\subset(0,1)$ that satisfies the boundary condition (\ref{in1}) at $y=1$.
\item [(ii)] If $\beta>\frac{\mu^2}{2\gamma \sigma^2}$, then $q_0^\beta(y)>0$ and $0>q_1^\beta(y)$ on their respective definition intervals.
\item [(iii)] Set $\beta=\frac{\mu^2}{2\gamma \sigma^2}-c$, with $c>0$. If $\lambda$ and $\varepsilon$ are sufficiently small, we have $q_0^\beta(y)<q_1^\beta(y)$ on $(0,y_0)\cap(y_1,1)$. Additionally, if $y_0<1$ then $\lim_{y\rightarrow y_0^-}q_0^\beta(y)=-\infty$, and if $y_1>0$ then $\lim_{y\rightarrow y_1^+}q_1^\beta(y)=+\infty$.
\end{itemize}
The solution to $q'=f(y,q)$ depends continuously on $\beta$. Hence, if $\lambda$ and $\varepsilon$ are sufficiently small, 
 we therefore have $q_0^{\beta^*}\equiv q_1^{\beta^*}$ for some $\beta^* \in \left[\max\{\mu-\gamma \sigma^2/2,0\}, \mu^2/(2\gamma \sigma^2) \right]$.\\

Proof of (i): This follows similarly as in \cite[Lemma A.8(i)]{guasoni.weber.13}. Indeed, replace $q^2/4\lambda$ by $(q-\varepsilon)^2/4\lambda$ in the first line of the proof of \cite[Lemma A.8(i)]{guasoni.weber.13}. Then, the proof proceeds analogously, leading to 
\begin{equation*}
h(0)=\varepsilon+2\sqrt{\lambda\beta}, \quad h'(0)=-\left(\frac{\lambda}{\beta}\right)^{1/2}(\mu+(\mu+\beta) h(0))-\varepsilon h(0)<0.
\end{equation*}
This term is negative like the corresponding expression in the first displayed equation in \cite[Lemma~A.8(i)]{guasoni.weber.13}. Hence, the remaining steps can be carried through unchanged.\\

Proof of (ii): This follows verbatim as in \cite[Lemma A.8(iii)]{guasoni.weber.13}.\\

Proof of (iii): Here, additional work is required compared to \cite{guasoni.weber.13}. The proof is based on the following observations: 

\begin{myrek}\label{remark1}
On $\{(y,q):y\in(0,1),q<\frac{1}{y}\}$, we have $f(y,q)\leq f_{NT}(y,q)$.
\end{myrek}

\begin{myrek}\label{remark2}
There are $\alpha,\zeta>0$, and $\eta<0$ independent of $\lambda$ and $\varepsilon$ such that 
$$f(y,q)\leq f_{NT}(y,q)\leq\eta<0, \quad \mbox{on } [y_\ast-\alpha,y_\ast+\alpha]\times [-\zeta,\zeta].$$
\end{myrek}

\begin{proof}[Proof of Remark~\ref{remark2}]
On $[y_-,y_+]$, Equation \eqref{eq:odeq} can be rewritten as 
$$y^2(1-y)^2 f_{NT}(y,q)=-c+k_1(y,q)+k_2(y,q),$$
where $\lim_{y\rightarrow y_\ast}k_1(y,q)=0$ and
$\lim_{q\rightarrow 0}k_2(y,q)=0$ uniformly. In particular, there is a negative constant $\eta$ such that, if $y$ is sufficiently close to $y_\ast$ and $q$ to $0$, we have $f_{NT}(y,q)\leq\eta$.
\end{proof}

A straightforward computation shows that one can choose $\lambda_1$ and $\varepsilon_1$ small enough such that, for any $\lambda\leq\lambda_1$ and $\varepsilon\leq\varepsilon_1$:
\begin{itemize}
\item $\frac{d}{dy}[y^2(1-y)^2f_{NT}(y,b_0(\lambda,\varepsilon))]<0$ on $(0,y_\ast-\alpha)$;\\
\item $\frac{d}{dy}[y^2(1-y)^2f_{NT}(y,b_1(\lambda,\varepsilon))]>0$ on $(y_\ast+\alpha,1)$;\\
\item  $\max\{b_0,-b_1\}<\zeta$ and $\eta<\frac{b_1(\lambda,\varepsilon)-b_0(\lambda,\varepsilon)}{2\alpha}<0$.
\end{itemize}
From the first point, and since $y^2(1-y)^2f_{NT}(y,b_0(\lambda,\varepsilon))=0$ for $y=0$, we get $f_{NT}(y,b_0(\lambda,\varepsilon))<0$ on $(0,y_\ast-\alpha)$. Remark~\ref{remark1} implies that $f(y,b_0(\lambda,\varepsilon))<0$ on $(0,y_\ast-\alpha)$.

The second point, with the same arguments, yields $f(y,b_1(\lambda,\varepsilon))<0$ on $(y_\ast+\alpha,1)$.

Consider now the line $q=\frac{b_1-b_0}{2\alpha}(y-y_\ast-\alpha)+b_1$. Given a solution $q(y)$, the third point and Remark~\ref{remark2} imply that if $q(y_\ast-\alpha)<b_0$ then $q(y)< \frac{b_1-b_0}{2\alpha}(y-y_\ast-\alpha)+b_1$ on $[y_\ast-\alpha,y_\ast+\alpha]$.
Define the following function:
\begin{equation}
g(y) = \left\{
        \begin{array}{ll}
            b_0, & \quad y\in(0,y_\ast-\alpha), \\
            \frac{b_1-b_0}{2\alpha}(y-y_\ast-\alpha)+b_1, & \quad y\in[y_\ast-\alpha,y_\ast+\alpha], \\
            b_1, & \quad y\in(y_\ast+\alpha,1).
        \end{array}
    \right.
\end{equation}
We have just shown that $f(y,g(y))<g'(y)$ on $(0,1)\setminus\{y_\ast-\alpha,y_\ast+\alpha\}$, thus $q_0^\beta(y)<g(y)$ on the definition interval of $q_0$. Analogously, $q_1^\beta(y)>g(y)$ and so $q_0^\beta(y)<q_1^\beta(y)$ on their common definition interval. This completes the proof of Item (iii).

We have shown that $q_0^\beta(y)>q_1^\beta(y)$ for $\beta>\frac{\mu^2}{2\gamma\sigma^2}$ and $q_1^\beta(y)>q_0^\beta(y)$ for $\beta=\frac{\mu^2}{2\gamma\sigma^2}-c$ with $c>0$ and sufficiently small $\varepsilon$ and $\lambda$ on their common domain. This proves that there exists $\beta^*= \frac{\mu^2}{2\gamma \sigma^2}-c(\varepsilon,\lambda)$ and $\bar y\in(0,1)$ such that $q_0^{\beta^*}(\bar y)=q_1^{\beta^*}(\bar y)$. Since both $q_0^{\beta^*}(\cdot)$ and $q_1^{\beta^*}(\cdot)$ satisfy equation~\eqref{eq:odeq}, $(q_0^{\beta^*})'(\bar y)=(q_1^{\beta^*})'(\bar y)$ and hence $q_0^{\beta^*}(y)=q_1^{\beta^*}(y)$ on the whole interval $(0,1)$. Therefore, we found a solution $q^{\beta^*}(y)$ satisfying conditions (\ref{in0}) and (\ref{in1}).
Finally, we show that the set $\{y: -\varepsilon< \frac{q^{\beta^*}(y)}{1-y q^{\beta^*}(y)} < \varepsilon \}$ is an interval $[y_-,y_+]$. In particular, it is enough to show that the solution $q^{\beta^*}(y)$ crosses the curves $q=\frac{\varepsilon}{1+\varepsilon y}$ and $q=-\frac{\varepsilon}{1-\varepsilon y}$ just once, in $y_-$ and $y_+$, respectively.

 The equation $f_{NT}(y,\frac{\varepsilon}{1+\varepsilon y})=\frac{d}{dy}(\frac{\varepsilon}{1+\varepsilon y})$ has exactly two solutions $y_2<y_3$ in $(0,1)$. In particular, $f_{NT}(y,\frac{\varepsilon}{1+\varepsilon y})>\frac{d}{dy}(\frac{\varepsilon}{1+\varepsilon y})$ on $(0,y_2)\cup(y_3,1)$ and $f_{NT}(y,\frac{\varepsilon}{1+\varepsilon y})<\frac{d}{dy}(\frac{\varepsilon}{1+\varepsilon y})$ on $(y_2,y_3)$. Since $q^{\beta^*}(0^+)>\varepsilon$ and $q^{\beta^*}(1^-)<\frac{\varepsilon}{1+\varepsilon}$, the solution $q^{\beta^*}(y)$ crosses $q=\frac{\varepsilon}{1+\varepsilon y}$ just once in $y_-\in(y_2,y_3)$. By the same arguments, $q^{\beta^*}(y)$ crosses $q=-\frac{\varepsilon}{1-\varepsilon y}$ just once in $y_+$.
\end{proof}

Unlike for small proportional transaction costs \cite{janecek.shreve.04,gerhold.al.14}, the frictionless Merton proportion $y_\ast$ does not generally lie in the no-trade region $[y_-,y_+]$ in the present setting. To see this, recall from Guasoni and Weber \cite[Remark A.13]{guasoni.weber.13} that in their model with price impact -- but without proportional transaction costs -- turnover is zero at exactly one point $y_\dagger$, which is $\mathcal{O}(\lambda^{1/2})$-close but not identical to $y_\ast$ for small $\lambda$. For a given small price impact $\lambda$, the trading boundaries $y_-,y_+$ converge to $y_\dagger$ as $\varepsilon \downarrow 0$. Hence, $y_\ast \notin [y_-,y_+]$ if the transaction cost is sufficiently small compared to the price impact. However, numerical evidence indicates that this effect only appears if the ratio $\varepsilon/\lambda$ is extremely small. Otherwise, the Merton proportion is contained in the no-trade region, compare Figures \ref{fig:tradingrates} and \ref{fig:asymptotics}.\\

 As shown in~\cite[Theorem 2.3]{guasoni.weber.13}, levered or short positions in the risky asset cannot be admissible with linear price impact. This remains true in the present setting with additional proportional costs:
 
\begin{mylemma}\label{ybound}
Let $u$ be an admissible strategy. Then, for sufficiently small $\varepsilon$ and $\lambda$, the corresponding risky weight
\begin{align}
dY_t &= (Y_t(1-Y_t)(\mu-Y_t \sigma^2) + (u(Y_t)+\varepsilon Y_t |u(Y_t)| +  \lambda Y_t u(Y_t)^2)) dt+ Y_t(1-Y_t) \sigma dW_t,\label{lemmadynamicsy}\\
Y_0 &= y \in (0,1),\label{initialvaluey}
\end{align}
takes values in $[0,1]$ a.s.\ for all $t$.
\end{mylemma}

\begin{proof}
With minor modifications, the assertion follows along the lines of~\cite[Lemma A.2, Lemma A.3, Lemma A.4, and Theorem 2.3]{guasoni.weber.13}.\footnote{In the proof of \cite[Lemma A.2]{guasoni.weber.13}, use $-(1-\varepsilon y)/4\lambda$ instead of $-1/4\lambda$ in the definition of $\tilde{\mu}$; then the proof can be carried through along the same lines. In the proof of Lemma \ref{ybound}, our analogue of \cite[Theorem 2.3]{guasoni.weber.13}, proportional transaction costs lead to an additional term $\int_0^T |\dot{\theta}_t|dt$ in the numerator of the expression analyzed in \cite[Lemma A.4]{guasoni.weber.13}. However, the latter still converges to zero by the same arguments as in the proof of \cite[Lemma A.4]{guasoni.weber.13}.} For the sake of completeness, we briefly recall the main ideas here. Let $u$ be any admissible strategy. First, verify that a stochastic process with the dynamics~\eqref{lemmadynamicsy} and initial value $y \in (1,\infty)$ resp. $y \in (-\infty,0)$ has a finite exploding time $\tau$ with positive 
probability, i.e., $\mathbb{P}[\tau < \infty] > 0.$ In a second step, show that the corresponding wealth process $X^u$ satisfies, $X^u_\tau (\omega)= 0$ a.s.\ on $\{\tau < \infty\}$ and $\theta_\tau > 0$ a.s. This in turn implies that any admissible strategy must fulfill $Y_t \in [0,1]$ a.s. for all $t$.
\end{proof} 
 
 Next, verify that the candidate strategy $\hat{u}$ from Theorem~\ref{mainresult} is admissible.

\begin{mylemma}\label{01uhat}
Define $\beta,q,y_-,y_+$ as in Lemma~\ref{odesolution} and set
$$
\hat{u}(y)=
\begin{cases}
\frac{1}{2\lambda}\left(\frac{q(y)}{1-y q(y)}-\varepsilon\right), &\text{if } y\in [0,y_-),\\
0, &\text{if } y \in [y_-,y_+],\\
\frac{1}{2\lambda}\left(\frac{q(y)}{1-y q(y)}+\varepsilon\right),  &\text{if } y\in (y_+,1].
\end{cases}
$$ 
Then, for sufficiently small $\varepsilon$ and $\lambda$, the SDE
\begin{align*}
dY^{\hat{u}}_t &= (Y^{\hat{u}}_t(1-Y^{\hat{u}}_t)(\mu-Y^{\hat{u}}_t \sigma^2) + (\hat{u}(Y^{\hat{u}}_t)+\varepsilon Y^{\hat{u}}_t |\hat{u}(Y^{\hat{u}}_t)| +  \lambda Y^{\hat{u}}_t \hat{u}(Y^{\hat{u}}_t)^2))dt+ Y^{\hat{u}}_t(1-Y^{\hat{u}}_t) \sigma dW_t,\\
Y^{\hat{u}}_0 &= y \in (0,1)
\end{align*}
has a unique strong solution which takes values in $[0,1]$ a.s.\ for all $t$. In particular, the strategy $\hat{u}$ is admissible.
\end{mylemma}

\begin{proof}
Lemma~\ref{odesolution} shows that $\hat{u}(y)$ is a bounded, continuous function on $[0,1]$ which satisfies $\hat{u} (0) >0$ and $\hat{u}(1) <0$, i.e., the strategy buys at full investment and sells at zero investment. Furthermore, notice that the scale function of the process $Y^{\hat{u}}$ is given by 
\begin{equation*}
s(x)= \int_c^x \exp{\left[-2 \int_c^y \frac{z(1-z)(\mu-z \sigma^2)+\hat{u}(z) + \varepsilon z |\hat{u}(z)|+ \lambda z \hat{u}(z)^2}{z^2 (1-z)^2\sigma^2} dz\right]} dy,
\end{equation*}
for $c\in(0,1)$. For sufficiently small $\varepsilon$ and $\lambda$, $\hat{u}(1)+ \varepsilon |\hat{u}(1)|+\lambda \hat{u}(1)^2 <0$. A straightforward computation shows that $s(0^+) = -\infty$ and $s(1^{-}) = \infty$, so that~\cite[Proposition 5.5.22]{karatzas.shreve.91} yields the first assertion. Finally, since $\hat{u}$ and $Y^{\hat{u}}$ are both bounded, the admissibility of the strategy follows.
\end{proof} 

With the function $q$, the constant $\beta$, and the boundaries $y_-,y_+$ of the no-trade region at hand, we can now use a variant of the verification argument\footnote{
The verification argument used in the proof was first used by~Guasoni and Robertson~\cite[Theorem 7]{guasoni.robertson.12} in a general Markovian frictionless setting. It was in turn adapted by Guasoni and Weber~\cite[Lemma A.7]{guasoni.weber.13} to a Black-Scholes model with quadratic trading costs. Here, we extend it to a setting with quadratic and linear trading costs, which leads to an additional nonlinearity in the HJB equation~\eqref{hjbcheck2}.} of Guasoni and Roberston~\cite[Theorem 7]{guasoni.robertson.12} to compute an upper bound for the equivalent safe rate of any admissible strategy:

\begin{mylemma}\label{esrbound}
Let $y\in (0,1)$ be the initial risky weight, define $\beta,q$ as in Lemma~\ref{odesolution}, and set $Q(\xi) = \int_0^\xi q(z) dz$. Then, the terminal wealth $X^u_T$ of any given admissible strategy $u$ satisfies:
\begin{equation}\label{eq:upperbound}
\mathbb{E}[(X^u_T)^{1-\gamma}]^{\frac{1}{1-\gamma}} \leq X_0 e^{\beta T + Q(y)} \mathbb{E}^{\hat{\mathbb{P}}^u} \left[e^{-(1-\gamma) Q(Y^u_T)}\right]^{\frac{1}{1-\gamma}},
\end{equation}
where 
\begin{equation}\label{eq:myopicprob}
\frac{d \hat{\mathbb{P}}^u|_{\mathcal{F}_T}}{d \mathbb{P}|_{\mathcal{F}_T}}= \mathcal{E} \left(\int_0^\cdot (1-\gamma) Y^u_s (1+q(Y^u_s)(1-Y^u_s)) \sigma dW_s\right)_T.
\end{equation}
Moreover, equality holds in \eqref{eq:upperbound} for the strategy $\hat{u}$ from Lemma~\ref{01uhat}.
\end{mylemma}

\begin{proof}
Fix an admissible strategy $u$ and omit the $u-$dependence of $X$, $Y$, and $\hat{\mathbb{P}}$ for the sake of clarity in the rest of the proof. Lemma~\ref{odesolution}, Lemma~\ref{ybound}, and Novikov's Condition imply that the stochastic exponential on the right-hand side of \eqref{eq:myopicprob} is a true martingale and therefore the density process of $\hat{\mathbb{P}}^u$ with respect to $\mathbb{P}$. 

Now, one readily checks that the assertion follows from
\begin{equation}\label{refor}
\log{X_T}-\log{X_0} -\frac{1}{1-\gamma} \log{\left(\frac{d \hat{\mathbb{P}}}{d \mathbb{P}}\right)} \leq \beta T -Q(Y_T)+ Q(y).
\end{equation}
To verify~\eqref{refor}, recall the dynamics of the wealth process $X$ and the risky weight $Y$ from (\ref{dynamicsX}-\ref{dynamicsY}) and apply It\^o's formula to $Q(Y_T)$ and $\log(X_T)$, obtaining: 
\begin{align}
Q(Y_T) - Q(y) =& \int_0^T q(Y_t)(Y_t(1-Y_t)(\mu-Y_t \sigma^2)+u_t + \varepsilon |u_t| Y_t+ \lambda Y_t u_t^2) dt\nonumber\\
 &  + \int_0^T \frac{1}{2}q'(Y_t) Y_t^2 (1-Y_t)^2 \sigma^2 dt+ \int_0^T q(Y_t)Y_t(1-Y_t) \sigma d W_t,\label{ito1}\\
\log{X_T}-\log{X_0} =& \int_0^T \left(Y_t \mu-\varepsilon |u_t|-\lambda u_t^2-\frac{1}{2}\sigma^2 Y_t^2\right) dt + \int_0^T \sigma Y_t d W_t.\label{ito2}
\end{align}
After substituting~(\ref{ito1}-\ref{ito2}) into~\eqref{refor}, this inequality reads as
\begin{align*}
&\int_0^T \mu Y_t-\varepsilon |u_t| -\lambda u_t^2-\frac{1}{2} \sigma^2 Y_t^2+ \frac{1}{2}\sigma^2 (1-\gamma) Y_t^2 (1+ q(Y_t)(1-Y_t))^2 dt\\
&\qquad \leq \int_0^T \left(\beta-q(Y_t)(Y_t(1-Y_t)(\mu-Y_t \sigma^2)+ u_t + \varepsilon |u_t|Y_t+ \lambda Y_t u_t^2)-\frac{1}{2} \sigma^2 Y_t^2(1-Y_t)^2q'(Y_t)\right) dt.
\end{align*}
Hence, it remains to verify that, for all $u\in \rr$ and $y \in [0,1]$:
\begin{equation}\label{hjbcheck1}
\begin{split}
&\mu y -\varepsilon |u| -\lambda u^2 -\frac{\sigma^2}{2} y^2+ \frac{1-\gamma}{2} \sigma^2 y^2 (1+ q(y)(1-y))^2 \\
&\qquad \leq \beta - q(y)(y(1-y)(\mu-y \sigma^2)+ u +\varepsilon |u| y+\lambda y u^2)-\frac{\sigma^2}{2}y^2(1-y)^2q'(y). 
\end{split}
\end{equation}
Rearranging~\eqref{hjbcheck1}, it suffices to check that, for all $u \in \mathbb{R}$ and $y \in [0,1]$:
\begin{equation}\label{hjbcheck2}
\begin{split}
0 &\geq -\beta+ \mu y -\frac{\gamma \sigma^2}{2}y^2+y (1-y) (\mu-\gamma \sigma^2 y) q + \frac{\sigma^2}{2} y^2 (1-y)^2 (q'+(1-\gamma) q^2)\\
 &  \qquad -\lambda u^2- \varepsilon |u| + (u+ \varepsilon |u| y + \lambda y u^2)q.
\end{split}
\end{equation}
Maximizing $-\lambda u^2- \varepsilon |u| + (u+ \varepsilon |u| y + \lambda y u^2)q$ over $u$ shows that the maximum\footnote{The condition $q(y)y < 1$ on $[0,1]$ guarantees that the critical point is indeed a maximum.} is attained at
\begin{equation*}
\tilde{u}(y)=
\begin{cases}
\frac{1}{2\lambda}\left(\frac{q(y)}{1-y q(y)}-\varepsilon\right), &\text{if } \frac{q(y)}{1-y q(y)}\geq\varepsilon,\\
0, &\text{if } -\varepsilon\leq\frac{q(y)}{1-y q(y)}\leq\varepsilon,\\
\frac{1}{2\lambda}\left(\frac{q(y)}{1-y q(y)}+\varepsilon\right), &\text{if } \frac{q(y)}{1-y q(y)}\leq-\varepsilon.
\end{cases}
\end{equation*}
In view of~\eqref{decreasingq},
\begin{equation*}
\tilde{u}(y)=
\begin{cases}
\frac{1}{2\lambda}\left(\frac{q(y)}{1-y q(y)}-\varepsilon\right), &\text{if } y\in [0,y_-),\\
0, &\text{if } y \in [y_-,y_+],\\
\frac{1}{2\lambda}\left(\frac{q(y)}{1-y q(y)}+\varepsilon\right), &\text{if } y\in (y_+,1].
\end{cases}
\end{equation*}
Now, the inequality \eqref{hjbcheck2} follows after substituting the ODE \eqref{eq:odeq} for $q$ and using the maximality of $\tilde{u}$. Evidently, this inequality becomes an equality for the strategy $\hat{u}$ from Lemma \ref{01uhat}.
\end{proof}

To complete the proof of Theorem~\ref{mainresult} we now verify that, as $T \rightarrow \infty$, the upper bound in Lemma~\ref{esrbound} converges to $\beta$ for any admissible strategy, and is attained for $\hat{u}$ from Lemma \ref{01uhat}. 

\begin{proof}[Proof of Theorem~\ref{mainresult}]
Let $\beta$ and $q$ be defined as in Lemma~\ref{odesolution} and let $u$ be an arbitrary admissible strategy. By Lemma~\ref{ybound}, we have $Y^u_t \in [0,1]$ for all $t$. As $q$ is bounded on $[0,1]$ due to Lemma~\ref{odesolution}, the function $Q(\xi) = \int_0^\xi q(z) dz$ is also bounded on $[0,1]$. Thus, for every admissible strategy, we have
\begin{equation*}
\lim_{T\rightarrow \infty} \frac{1}{(1-\gamma)T}\log{\mathbb{E}^{\hat{\mathbb{P}}^u}}[e^{-(1-\gamma) (Q(Y^u_T)-Q(y))}]=0.
\end{equation*}
As $T\rightarrow \infty$, Lemma~\ref{esrbound} therefore provides a strategy-\emph{independent} upper bound for the equivalent safe rate:
\begin{equation*}
\esr_\gamma (u) = \lim_{T\rightarrow \infty} \frac{1}{(1-\gamma) T}\log{\mathbb{E}[(X^u_T)^{1-\gamma}]} \leq \beta.
\end{equation*}
This upper bound is attained for the admissible strategy $\hat{u}$ from Lemma~\ref{01uhat}. Hence, the latter is long-run optimal with equivalent safe rate $\beta$ as claimed.
\end{proof}

\subsection{Proof of Propositions~\ref{limit_esr}, \ref{prop_asymptotics}, and \ref{prop_convergence}}\label{proof:prop_asymptotics}

In this final section, we provide proofs for the asymptotic results from Section~\ref{sec:asymptotics}. 

\begin{mylemma}\label{q_decreasing}
Assume $\lambda=K\varepsilon^{4/3}$. For sufficiently small $\varepsilon$, the solution $q(y)$ defined in Theorem~\ref{mainresult} is strictly decreasing.
\end{mylemma}

\begin{proof}
Rewrite equation~\eqref{eq:odeq} as $q'(y)=f(y,q(y))$. In the proof of Theorem~\ref{mainresult} we show that there exists a function $h(y)$ defined close to $0$ such that $f(y,h(y))=0$, $h(0^+)=b_0$ and $h'(0^+)<0$. Notice that $f(y,\varepsilon\pm 2\sqrt{K\beta}\varepsilon^{2/3})<0$, therefore $|h(y)-\varepsilon|\leq2\sqrt{K\beta}\varepsilon^{2/3}$, and in particular $h(y)<\frac{1}{y}$ for small $\varepsilon$. The equation $f(y,h)=0$ can then be seen as a cubic equation in $h$ (that reduces to a quadratic equation when $h<\frac{\varepsilon}{1+\varepsilon y}$), and there exists $N_0$ such that eventually (when $\varepsilon$ goes to $0$) this equation has positive discriminant on $(0,y_*-N_0 \sqrt{c(\varepsilon)})$, where $c(\varepsilon):=\frac{\mu^2}{2\gamma\sigma^2}-\beta$ (and discriminant equal to $0$ at $y_*-N_0 \sqrt{c(\varepsilon)}$). In this case there are explicit expressions for the three real roots of $f(y,h)=0$ in terms of $y$, one of these roots extends the function $h(y)$ to a continuous function on the whole interval $(0,y_*-N_0 \sqrt{c(\varepsilon)})$. Notice that on $(y_*-N_0 \sqrt{c(\varepsilon)},y_*)$ the cubic equation has only one real root that does not lie in $(b_1,b_0)$, in other terms $f(y,q)<0$ on $(y_*-N_0 \sqrt{c(\varepsilon)},y_*)\times(b_1,b_0)$.\\
Using the fact that $|h(y)-\varepsilon|\leq2\sqrt{K\beta}\varepsilon^{2/3}$, we see that there does not exist any $y\in(0,y_*-N_0 \sqrt{c(\varepsilon)})$ such that simultaneously $\frac{d}{dy}f(y,h(y))=0$ and $h'(y)=0$. This implies that $h(y)$ is monotone on its whole definition interval, in particular (since $h'(0^+)<0$) it is a decreasing function.\\
The solution $q(y)$ to~\eqref{eq:odeq} defined in Theorem~\ref{mainresult} is such that $q(0^+)=b_0$ and $q'(0^+)=0$, while $h(0^+)=b_0$ and $h'(0^+)<0$, therefore close to $0$ we have $q>h$. Since $h$ is a subsolution, i.e., $0=f(y,h(y))>h'(y)$, we have $q(y)>h(y)$ on $(0,y_*-N_0 \sqrt{c(\varepsilon)})$. This implies that $q'(y)=f(y,q(y))<0$ on $(0,y_*-N_0 \sqrt{c(\varepsilon)})$.\\
A similar argument proves that $q'(y)<0$ on $(y_*+N_1 \sqrt{c(\varepsilon)},1)$ for some $N_1>0$. To conclude, it is enough to observe that $f(y,q)<0$ on $(y_*-N_0 \sqrt{c(\varepsilon)},y_*+N_1 \sqrt{c(\varepsilon)})\times(b_1,b_0)$.
\end{proof}

Proposition \ref{limit_esr} shows that the equivalent safe loss, i.e., the positive value $c(\varepsilon):=\frac{\mu^2}{2\gamma\sigma^2}-\esr_\gamma(\hat u)$, is of order $\varepsilon^{2/3}$. The lower bound for $c(\varepsilon)$ is obtained by comparing it to the loss that would occur if the investor were facing only quadratic costs, while the upper bound is obtained by adapting arguments for a model with nonlinear trading costs of power form~\cite{guasoni.weber.nonlinear}:

\begin{proof}[Proof of Proposition~\ref{limit_esr}]

For any admissible strategy $u$, let $\esr^{q.c.}_\gamma(u)$ be the equivalent safe rate for the model with quadratic costs only as in \cite{guasoni.weber.13}. It is clear that for any admissible strategy $u$, $\esr^{q.c.}_\gamma(u)\geq\esr_\gamma(u)$. In~\cite{guasoni.weber.13} it is proved that the maximal equivalent safe rate in the model with only quadratic cost is of the form $\frac{\mu^2}{2\gamma\sigma^2}-\varepsilon^{2/3}l^{q.c.} + o(\varepsilon^{2/3})$ for some positive constant $l^{q.c.}$. This implies that $\liminf_{\varepsilon\rightarrow0}\frac{c(\varepsilon)}{\varepsilon^{2/3}}\geq l^{q.c}>0$.

We will now prove that $\limsup_{\varepsilon\rightarrow0}\frac{c(\varepsilon)}{\varepsilon^{2/3}}<+\infty$. To this end, we assume that there exists $\varepsilon_n\downarrow0$ such that $\lim_{n\rightarrow\infty}\frac{c(\varepsilon_n)}{\varepsilon_n^{2/3}}=+\infty$ and and show that this leads to a contradiction.

Let $q(\cdot)$ be the solution to equation~\eqref{eq:odeq} as in Theorem~\ref{mainresult}. Define $y_l^n:=y_*-\sqrt{\frac{2}{\gamma\sigma^2}(1+\delta)c(\varepsilon_n)}$ for some fixed (i.e., independent of $\varepsilon$) and positive $\delta$ and let $\tilde y^n$ be the only point such that $q(\tilde y^n)=0$. Assume that, up to a subsequence, $\tilde y^n\geq y_*$. (If such a subsequence does not exist, the same argument with $y_r^n:=y_*+\sqrt{\frac{2}{\gamma\sigma^2}(1+\delta)c(\varepsilon_n)}$ can be used.) We divide the rest of the proof into two parts:\\

(i) First we prove that the assumption $\lim_{n\rightarrow\infty}\frac{c(\varepsilon_n)}{\varepsilon_n^{2/3}}=+\infty$ implies that there exists a subsequence $n_k$ such that
\begin{equation}\label{2.1}
\lim_{k\rightarrow\infty}\frac{(q(y_l^{n_k})-\varepsilon)^2}{\varepsilon^{4/3}_{n_k}c(\varepsilon_{n_k})}=+\infty.
\end{equation}
Recall that $q(0^+)=\varepsilon+2\sqrt{K\beta}\varepsilon^{2/3}=:b_0(\varepsilon)$ and that $q(\cdot)$ is decreasing and therefore $q(y)\leq b_0$ on $[0,1]$. If there is $C_1>1+2\sqrt{K\frac{\mu^2}{2\gamma\sigma^2}}$ independent of $\varepsilon$ such that eventually $q(y_l^n)\geq C_1 c(\varepsilon_n)$, then to conclude it is enough to note that for large $n$ eventually $b_0(\varepsilon_n)<C_1 \varepsilon_n^{2/3}<C_1 c(\varepsilon_n)\leq q(y_l^n)$, where the second inequality follows from the assumption that $\lim_{n\rightarrow\infty}\frac{c(\varepsilon_n)}{\varepsilon_n^{2/3}}=+\infty$. This contradicts the monotonicity of $q(\cdot)$, thus we can assume (up to a subsequence) that $q(y_l^n)< C_1 c(\varepsilon_n)$.
With this assumption, and also using the monotonicity of $q(y)$ established in Lemma~\ref{q_decreasing}, on the interval $( y_*-\sqrt{\frac{2}{\gamma\sigma^2}(1-\delta)c(\varepsilon)}, y_*)$ we obtain from Equation~\eqref{eq:odeq} that, for some constant $C_2$,
\begin{equation*}
\frac{\sigma^2}{2}y^2(1-y)^2 q'(y)<(1-\delta)c(\varepsilon)-c(\varepsilon)+0+C_2 c(\varepsilon)^2+0.
\end{equation*}
Thus, there is $C_3>0$ such that $q'(y)<-C_3 c(\varepsilon)$ on $( y_*-\sqrt{\frac{2}{\gamma\sigma^2}(1-\delta)c(\varepsilon)}, y_*)$ for sufficiently small $\varepsilon$. As a consequence, it follows from the definition of of $y_l^n$ and the assumption $\tilde y^n\geq y_*$ that
\begin{align*}
q(y_l^n)&=-\int_{y_l^n}^{\tilde y^n}q'(y)dy\geq-\int_{ y_*-\sqrt{\frac{2}{\gamma\sigma^2}(1-\delta)c(\varepsilon)}}^{ y_*}q'(y)dy\\&>\int_{ y_*-\sqrt{\frac{2}{\gamma\sigma^2}(1-\delta)c(\varepsilon)}}^{ y_*}C_3 c(\varepsilon)dy=C_3 \sqrt{\frac{2}{\gamma\sigma^2}(1-\delta)}c(\varepsilon)^{3/2}.
\end{align*}
Recall that $y_-$ is by definition the only point such that $q(y_-)=\frac{\varepsilon}{1+\varepsilon y_-}$. Since $\frac{c(\varepsilon_n)^{3/2}}{\varepsilon_n}\uparrow\infty$ and therefore $q(y_l^n)>\frac{\varepsilon}{1+\varepsilon y_l^n}$, this implies that eventually $y_l^n<y_-$. Furthermore we get that, for some $C_4$ and $C_5$,
\begin{equation*}
 \frac{(q(y_l^n)-\varepsilon_n)^2}{\varepsilon^{4/3}_n c(\varepsilon_n)}\geq \frac{(C_4 c(\varepsilon_n)^{3/2}-\varepsilon_n)^2}{\varepsilon^{4/3}_n c(\varepsilon_n)}\geq \frac{C_5 c(\varepsilon_n)^3}{\varepsilon^{4/3}_n c(\varepsilon_n)}= C_5 \left(\frac{c(\varepsilon_n)}{\varepsilon_n^{2/3}}\right)^2\uparrow\infty,
\end{equation*}
which proves point (i).\\

(ii) We next show that the assumption that $\lim_{n\rightarrow\infty}\frac{c(\varepsilon_n)}{\varepsilon_n^{2/3}}\uparrow\infty$, together with \eqref{2.1} implies that, for some $\eta\in(0, y_l^n)$ and sufficiently small $\varepsilon_n$,
\begin{equation*}
q(\eta)>q(0^+)=b_0(\varepsilon_n).
\end{equation*}
This contradicts the monotonicity of $q(y)$, thereby producing the desired contradiction.

Fix a large constant $M>0$. Since $\lim_{n\rightarrow\infty} \frac{(q(y_l^n)-\varepsilon_n)^2}{\varepsilon_n^{4/3}c(\varepsilon_n)}=+\infty$, for sufficientely small $\varepsilon_n$ we have
\begin{equation*}
\frac{(q(y_l^n)-\varepsilon_n)^2}{\varepsilon_n^{4/3}}>M\delta c(\varepsilon_n)=M\left(\frac{\gamma\sigma^2}{2}(y_l^n- y_*)^2-c(\varepsilon_n)\right).
\end{equation*}
Define $y_{1,n}:=\inf\{y\in(0, y_*):-\frac{\sigma^2}{2}y^2(1-y)^2(1-\gamma)q^2(y)>\frac{\gamma\sigma^2}{2}(y- y_*)^2-c(\varepsilon_n)\}>0$, with the usual convention $\inf\emptyset=+\infty$. Assume that eventually $y_{1,n}<y_l^n$. Define $p_n:=-\frac{\sigma^2}{2}y_{1,n}^2(1-y_{1,n})^2(1-\gamma)$; we want to show that also in $y_{1,n}$ this inequality holds:
\begin{equation*}
\frac{(q(y_{1,n})-\varepsilon_n)^2}{\varepsilon_n^{4/3}}>M\left(\frac{\gamma\sigma^2}{2}(y_{1,n}- y_*)^2-c(\varepsilon_n)\right)=p_n M q^2(y_{1,n}),
\end{equation*}
where we used the definition of $y_{1,n}$ in the last equality. This inequality is trivially satisfied if $p_n<0$ (i.e., when $\gamma<1$), so assume $p_n$ is positive. Recall from part (i) of the proof that $q(y_{1,n})>q(y_l^n)\geq C_4 c(\varepsilon_n)^{3/2}$ and therefore that eventually $q(y_{1,n})>\varepsilon$. Taking square roots (using $q(y_{1,n})>\varepsilon>0$) and rearranging the terms in the inequality to be proved, we obtain the equivalent inequality $q(y_{1,n})>\frac{\varepsilon_n}{1-\sqrt{p_n M}\varepsilon_n^{2/3}}$. Since $q(y_{1,n})> C_4 c(\varepsilon_n)^{3/2}$, it is sufficient to show that $C_4 c(\varepsilon_n)^{3/2}\geq\frac{\varepsilon_n}{1-\sqrt{p_n M}\varepsilon^{2/3}}$. This can be rewritten as $C_4 \left(\frac{c(\varepsilon_n)}{\varepsilon_n^{2/3}}\right)^{3/2}\geq\frac{1}{1-\sqrt{p_n M}\varepsilon^{2/3}}$, which holds true for small $\varepsilon_n$ since $\frac{c(\varepsilon_n)}{\varepsilon_n^{3/2}}\uparrow\infty$ and $p_n$ is uniformly bounded in $\varepsilon_n$.

If eventually $y_{1,n}<y_l^n$, let $y_{2,n}$ be equal to $y_{1,n}$. Otherwise, up to a subsequence, $y_l^n<y_{1,n}$ and then define $y_{2,n}:=y_l^n$. We have just shown that in both cases
\begin{equation*}
\frac{(q(y_{2,n})-\varepsilon_n)^2}{\varepsilon_n^{4/3}}>M\left(\frac{\gamma\sigma^2}{2}(y_{2,n}- y_*)^2-c(\varepsilon_n)\right)
\end{equation*}
and $y_{2,n}\leq y_{1,n}$.

Choose $\eta$ such that $\eta<y_{2,n}$ for every $n$ (this is possible because for $y\downarrow0$, $-\frac{\sigma^2}{2}y^2(1-y)^2(1-\gamma)q^2(y)\downarrow0$ and $q(y)$ is uniformly bounded in $\varepsilon$ and $\frac{\gamma\sigma^2}{2}(y- y_*)^2-c(\varepsilon_n)\uparrow\beta>0$). On the interval $(\eta,y_{2,n})$ we have
\begin{multline}\label{ineqtilde}
\frac{\sigma^2}{2}\frac{1}{16}q'(y)\leq\frac{\sigma^2}{2}y^2(1-y)^2q'(y)\\<\frac{\gamma\sigma^2}{2} (y- y_*)^2-c(\varepsilon_n)+0-\frac{\sigma^2}{2}y^2(1-y)^2(1-\gamma)q^2(y)-\frac{(q(y)-\varepsilon_n)^2}{4K\varepsilon_n^{4/3}}\\
\leq 2\left[\frac{\gamma\sigma^2}{2} (y- y_*)^2-c(\epsilon_n)\right]-\frac{(q(y)-\varepsilon_n)^2}{4K\varepsilon_n^{4/3}},
\end{multline}
where the first inequality follows from $y^2(1-y)^2\leq\frac{1}{16}$, the second from $y< y_*$ and $q>0$, the third from $y_{2,n}\leq y_{1,n}$ and the definition of $y_{1,n}$. 
Define $\tilde q(y)$ on $(\eta,y_{2,n})$ as the solution of the Cauchy problem
\begin{align*}
\tilde q'(y)&=\tilde f(y,\tilde q(y))\\:&=\frac{32}{\sigma^2} \left(2\left[\frac{\gamma\sigma^2}{2} (y- y_*)^2-c(\varepsilon_n)\right]-\frac{(q(y)-\varepsilon_n)^2}{4K\varepsilon_n^{4/3}}\right),\\
\tilde q(y_{2,n})&=q(y_{2,n}).
\end{align*}
The inequality \eqref{ineqtilde} implies that $q'(y)<\tilde f(y,q(y))$. Thus, $\tilde q(y)<q(y)$ on $(\eta,y_{2,n})$.
Define on $(\eta,y_{2,n})$ the function $k(y)>\varepsilon_n$ by
\begin{equation*}
\frac{(k(y)-\varepsilon_n)^2}{\varepsilon_n^{4/3}}=M\left(\frac{\gamma\sigma^2}{2} (y- y_*)^2-c(\varepsilon_n)\right).
\end{equation*}
To conclude the proof it remains to show that $k(y)$ is a subsolution for $\tilde q(y)$ on $(\eta,y_{2,n})$, which implies that $k(\eta)<q(\eta)$. In particular, if $M$ is chosen sufficiently large, we get the contradiction.

Thus, it remains to show that $\tilde f(y,k(y))<k'(y)$ on $(\eta,y_{2,n})$, i.e.,
\begin{multline*}
\frac{32}{\sigma^2} \left(2-\frac{M}{4K}\right)\left[\frac{\gamma\sigma^2}{2} (y- y_*)^2-c(\varepsilon_n)\right]< \frac{\sqrt{M}}{2}\varepsilon_n^{2/3}\left[\frac{\gamma\sigma^2}{2} (y- y_*)^2-c(\varepsilon_n)\right]^{-1/2}\gamma\sigma^2(y-y_*).
\end{multline*}
This is equivalent to
\begin{equation}\label{eq:lemmaesr}
\frac{64}{\gamma\sigma^4} \left(\frac{\sqrt{M}}{4K}-\frac{2}{\sqrt{M}}\right)> \varepsilon_n^{2/3}\left[\frac{\gamma\sigma^2}{2} (y- y_*)^2-c(\varepsilon_n)\right]^{-3/2}(y_*-y)=:r(y).
\end{equation}
A simple calculation shows that $r'(y)>0$ on $(0,y_l^n)$. Thus, $r(y)\leq r(y_l^n)$ on $(\eta,y_{2,n})$, where
\begin{equation*}
r(y_l^n)=\varepsilon_n^{2/3}(\delta c(\varepsilon_n))^{-3/2}\left(\frac{2}{\gamma\sigma^2}(1+\delta)c(\varepsilon_n)\right)^{1/2}=D_1 \frac{\varepsilon_n^{2/3}}{c(\varepsilon)}\downarrow0,
\end{equation*}
with $D_1>0$. If $M$ is large enough, the left-hand side in (\ref{eq:lemmaesr}) is strictly positive. Since $r(y)$ converges to $0$ on $(\eta,y_{2,n})$ uniformly in $\varepsilon$, for sufficiently small $\varepsilon$ inequality (\ref{eq:lemmaesr}) holds true. This completes the proof.
\end{proof}

Next, we establish existence and uniqueness for the inhomogeneous Riccati equation that determines the small-cost asymptotics in 
Section~\ref{sec:asymptotics}. To this end, we first prove an auxiliary result about Riccati ODEs:

\begin{mylemma}\label{lemma_riccati}
The Riccati equation
\begin{equation}\label{lemma_eq}
y'(x)=f(x,y(x)):=-ax^2+b+cy^2(x),
\end{equation}
with $a,b,c>0$ has a unique solution such that
\begin{equation}\label{lemma_growth}
\lim_{x \rightarrow \infty} \frac{y(x)}{\sqrt{a/c} x}=1.
\end{equation}
Furthermore, in equation \eqref{lemma_eq} replace the parameter $b$ with respectively $b_1$ and $b_2$, and consider the corresponding unique solutions $y_1(x)$ and $y_2(x)$ that satisfy \eqref{lemma_growth}. If $b_1<b_2$, then $y_1(x)>y_2(x)$.
\end{mylemma}

\begin{proof}
On $(\sqrt{b/a},+\infty)$, define the function $h(x):=\sqrt{a x^2/c-b/c}$. Notice that by definition of $h(x)$ we have $f(x,h(x))=0$. For each $\bar x\in (\sqrt{b/a},+\infty)$ consider the solution $y(x;\bar x, h(\bar x))$ with initial condition $(\bar x, h(\bar x))$ and define $y_*(x):=\sup\{y(x;\bar x,h(\bar x)):\bar x\in (\sqrt{b/a},+\infty)\}$.

For any $x_1$ there is a large $y_1$ such that the linear function $\tilde y(x)=y_1+\sqrt{a/c}(x-x_1)$ is a subsolution to \eqref{lemma_eq} whose graph does not intersect the graph of $h(x)$. In particular, the solution $y(x;x_1,y_1)$ to \eqref{lemma_eq} with initial condition $(x_1,y_1)$ is (on its definition interval) strictly larger than $h(x)$. Since for any $\bar x$ in its definition interval $y(\bar x;x_1,y_1)> h(\bar x)=y(\bar x;\bar x, h(\bar x))$, we have also that $y_1=y(x_1;x_1,y_1)> y(x_1;\bar x, h(\bar x))$ and thus $+\infty>y_1\geq y_*(x_1)$. This argument can be repeated for any $x_1\in(\sqrt{b/a},+\infty)$, hence $y_*(x)<+\infty$ on $(\sqrt{b/a},+\infty)$.

We want to prove that $y_*(x)$ is the unique solution that satisfies \eqref{lemma_growth}. By construction, $y_*(x)$ has the following properties:
\begin{itemize}[i)]
\item[i)] $y_*(x)\geq h(x)$;
\item[ii)] $\left(\sqrt{\frac{b}{a}},+\infty\right)\subset D$, where $D$ is the domain of $y_*(x)$.
\end{itemize}
From property (i) it follows that $\liminf_{x \rightarrow \infty} \frac{y_*(x)}{\sqrt{a/c} x}\geq 1.$

Next we show that $L:=\lim_{x \rightarrow \infty} \frac{y_*(x)}{\sqrt{a/c} x}$ exists. Assume $\limsup_{x \rightarrow \infty} \frac{y_*(x)}{\sqrt{a/c} x}=:M\in(1,+\infty)$ (the case $M=+\infty$ is analogous). Then, there is a sequence $(x_n)_{n\geq0}$ such that $\lim_{n \rightarrow \infty} \frac{y_*(x_n)}{\sqrt{a/c} x_n}=M$. In particular, for any $\delta\in(0,M-1)$ there exists $N_\delta \in \mathbb{N}$ such that $\forall n \geq N_\delta$ we have $y_*(x_n)\geq(M-\delta)\sqrt{a/c} x_n$. For large $x$, the function $s(x)=(M-\delta)\sqrt{a/c} x$ is a subsolution to \eqref{lemma_eq}, because
\begin{equation*}
(M-\delta)\sqrt{\frac{a}{c}}=s'(x)\leq -ax^2+b+cs^2(x)=ax^2((M-\delta)^2-1)+b.
\end{equation*}
Thus, for every $\delta\in(0,M-1)$ and some $\bar x$, we have $y_*(x)\geq(M-\delta)\sqrt{a/c} x$ for $x\geq\bar x$. In particular, $\liminf_{x \rightarrow \infty} \frac{y_*(x)}{\sqrt{a/c} x}\geq M-\delta$ for any small $\delta$, and $\liminf_{x \rightarrow \infty} \frac{y_*(x)}{\sqrt{a/c} x}=M=\limsup_{x \rightarrow \infty} \frac{y_*(x)}{\sqrt{a/c} x}$. In other terms, the limit $L$ exists.
 
We prove next that $L=1$. First, assume by contradiction that $1<L<+\infty$. 
Since $\lim_{x \rightarrow \infty} \frac{y_*(x)}{\sqrt{a/c} x}=L<+\infty$, the function $y_*(x)$ grows linearly. On the other hand 
from \eqref{lemma_eq} one gets
\begin{equation*}
\lim_{x\rightarrow\infty}\frac{y_*'(x)}{a x^2}=L-1>0,
\end{equation*}
which implies that $y'_*(x)$ grows quadratically, leading to a contradiction.

Assume now that $L=+\infty$. From \eqref{lemma_eq} it follows that
\begin{equation*}
\lim_{x\rightarrow\infty}\frac{y_*'(x)}{c y_*^2(x)}=1.
\end{equation*}
For small $\delta$ and sufficiently large $x$, we have $(1-\delta)c y_*^2(x)\leq y_*'(x)$. This implies that $y_*(x)$ is bounded from below by a positive function of the form $\frac{1}{k-(1-\delta)cx}$ for some $k>0$. In particular, $y_*(x)$ would have a vertical asymptote, contradicting property (ii). This proves that $L=1$.\\

The next step is to prove uniqueness. Consider a sufficiently small $\delta>0$ and $\bar x$ such that for any $x\geq\bar x$:
\begin{equation*}
\frac{1}{\sqrt{ac}x^2}\leq\delta \qquad \text{and}\qquad \frac{y_*(x)}{\sqrt{a/c} x}\geq 1-\delta.
\end{equation*}
For any $d>0$, consider the function $w(x;d)=y_*(x)+dx$. We now show that for $x\geq\bar x$, $w(x;d)$ is a subsolution to \eqref{lemma_eq}, i.e., $w'(x;d)\leq f(x,w(x;d))$. Since $y_*(x)$ is a solution to \eqref{lemma_eq}, this inequality is equivalent to $d\leq c d^2x^2+2 cdy_*(x)x$, which can be rearranged to
\begin{equation*}
\frac{1}{\sqrt{ac}x^2}\leq \sqrt{\frac{c}{a}} d+2 \sqrt{\frac{c}{a}}\frac{y_*(x)}{x}.
\end{equation*}
Since $x\geq\bar x$, this inequality follows from $\delta\leq \sqrt{\frac{c}{a}} d+2 -2\delta$, given that $\delta$ was chosen appropriately. Thus, $w(x;d)$ is a subsolution for any $d>0$. In particular, let $y_2(x)>y_*(x)$ be a solution to \eqref{lemma_eq} and choose $d_*$ such that $y_2(\bar x)=y_*(\bar x)+d_*\bar x$. Then $y_2(x)\geq w(x;d_*)=y_*(x)+d_* x$ for $x\geq \bar x$ and $y_2(x)$ cannot satisfy \eqref{lemma_growth}. Since any solution smaller than $y_*(x)$ is also -for large $x$- smaller than $h(x)$ and thus eventually decreasing, this is enough to prove uniqueness.

Finally, define as before $h_1(x)=\sqrt{ax^2/c-b_1/c}$ and $h_2(x)=\sqrt{ax^2/c-b_2/c}$. Since $h_1(x)> h_2(x)$ and since any solution to equation \eqref{lemma_eq} with coefficient $b_1$ is a subsolution for equation \eqref{lemma_eq} with coefficient $b_2$, the solution $y_1(x,\bar x)$ to the first equation with initial condition $(\bar x,h_1(\bar x))$ is above the solution $y_2(x,\bar x)$ to the second equation with initial condition $(\bar x,h_2(\bar x))$ on $(-\infty,\bar x]$. It follows that $y_1(x)>y_2(x)$.
\end{proof}

\begin{proof}[Proof of Proposition~\ref{prop_asymptotics}]
For any $l>0$, define $r_B(z;l)$ as the unique solution of \eqref{leftr} that satisfies 
\begin{equation}\label{as_growth}
\lim_{z \rightarrow -\infty} \frac{r_B(z)}{-\sqrt{2 K \gamma \sigma^2} z}=1. 
\end{equation}
Here, existence and uniqueness of this solution follow from Lemma \ref{lemma_riccati}. Let $r(z;l)$ be the unique solution to \eqref{midr} with initial condition $r(0;l)=0$. It is enough to prove that there exists a unique $l$ such that for some $z_-<0$ we get $r_B(z_-;l)=r(z_-;l)=1$, $r_B(z_-;l)>1$ on $(-\infty,z_-)$ and $r(z_-;l)<1$ on $(z_-,0]$.

Define $l_\lambda:=\sqrt{\frac{\gamma K}{2}}\sigma^3y_\ast^2(1-y_\ast)^2$ and $l_\varepsilon:=\left(\frac{3}{4}\sqrt{\frac{\gamma}{2}}\sigma^3y_\ast^2(1-y_\ast)^2\right)^{2/3}$. If the following values exist, define $z_B(l)$ the (only) value such that $r_B(z_B(l);l)=1$ and $r'_B(z_B(l);l)<0$, and $z_{NT}(l)$ the (only) value such that $r(z_{NT}(l);l)=1$ and $r'(z_{NT}(l);l)<0$.\\
If $l=l_\lambda$, the solution to ODE~\eqref{leftr} with boundary condition~\eqref{as_growth} is $r_B(z;l_\lambda)=-\sqrt{2\gamma K}\sigma z+1$. Thus, $z_B(l_\lambda)=0$. From Lemma \ref{lemma_riccati} it follows that if $l_1<l_2$, then $r_B(z;l_1)>r_B(z;l_2)$ and thus $z_B(l_1)>z_B(l_2)$. Hence, the function $l\rightarrow z_B(l)$ is well-defined and decreasing on $[l_\lambda,+\infty)$. Furthermore, $\lim_{l\rightarrow\infty}z_B(l)=-\infty$.

Since $\max_{z\in(-\infty,0)}r(z;l)<1$ if and only if $l<l_\varepsilon$, the function $l\rightarrow z_{NT}(l)$ is well-defined only on $[l_\varepsilon,+\infty)$. Furthermore, it is increasing and $\lim_{l\rightarrow\infty}z_{NT}(l)=0$.

Let $l_M:=\max\{l_\lambda,l_\varepsilon\}$. If $l_\lambda\geq l_\varepsilon$, then $z_{NT}(l_\lambda)< z_B(l_\lambda)=0$.

A brief calculation shows that $z_{NT}(l_\varepsilon)=-\sqrt{\frac{2l_\varepsilon}{\gamma\sigma^2}}$. Thus, for $(z,r)\in(-\infty,z_{NT}(l_\varepsilon))\times\{1\}$ any solution to \eqref{leftr} has strictly positive derivative, while it is strictly negative for $(z,r)\in(z_{NT}(l_\varepsilon),0]\times\{1\}$. In particular, since $r'_B(z_B(l);l)<0$, if $l_\varepsilon\geq l_\lambda$, then $z_B(l_\varepsilon)\in(z_{NT}(l_\varepsilon),0]$. In both cases, $z_{NT}(l_M)< z_B(l_M)$. Given the monotonicity of the functions $z_{NT}(l)$ and $z_B(l)$, there exists a unique $l_*\in(l_M,+\infty)$ such that $z_{NT}(l_*)= z_B(l_*)$.

Without price impact the no-trade region is given by $(z_{NT}(l_\varepsilon),-z_{NT}(l_\varepsilon))$ (see ~\cite[Formula (2.9)]{gerhold.al.14}). Since $z_{NT}(l)$ is an increasing function, the no-trade region with price impact is strictly smaller: $(z_{NT}(l_*),-z_{NT}(l_*))\subset(z_{NT}(l_\varepsilon),-z_{NT}(l_\varepsilon))$.
\end{proof}

Finally, we rigorously prove that the value function $q^\varepsilon(y)$ converges to the solution $r(z)$ specified in Proposition~\ref{prop_asymptotics}.

\begin{proof}[Proof of Proposition~\ref{prop_convergence}]
Rewrite equation~\eqref{eq:odeq} as $q'(y)=f^\varepsilon(y,q(y))$ and equation (\ref{leftr}-\ref{rightr}) as $r'(z)=g(z,r(z))$. Consider a sequence $(\varepsilon_n)_{n\geq0}$ such that $\frac{c(\varepsilon_n)}{\varepsilon_n^{2/3}}$ has limit equal to $l$ (for the rest of the proof we will write $\varepsilon$ instead of $\varepsilon_n$ for simplicity), such a limit exists by Proposition~\ref{limit_esr}. In the proof of Theorem~\ref{mainresult} we show that there exists a function $h^\varepsilon(y)$ defined close to $0$ such that $f^\varepsilon(y,h^\varepsilon(y))=0$, $h^\varepsilon(0^+)=b_0(\varepsilon)$ and $(h^\varepsilon)'(0^+)<0$. In the proof of Lemma~\ref{q_decreasing} we prove that $h^\varepsilon$ is decreasing on its whole definition interval $(0,y_*-M\varepsilon^{1/3})$.

In view of the continuity guaranteed by the implicit function theorem, for $\varepsilon\downarrow0$, the function $\varepsilon^{-1}h^\varepsilon(\varepsilon^{1/3}z+y_*)$ converges to $k(z):=1+2\sqrt{K\left(\frac{\gamma\sigma^2}{2}z^2-l\right)}$, i.e., $k(z)$ such that $g(z,k(z))=0$.\\
In the proof of Theorem~\ref{mainresult} we show that there exists a unique solution $q_0^\varepsilon$ to equation~\eqref{eq:odeq} with limit $q^\varepsilon(0^+)=b_0(\varepsilon)$. Define $q^\varepsilon(y;\bar y, \bar q)$ as the solution to equation~\eqref{eq:odeq} with initial condition $(\bar y,\bar q)$ and $\tilde q_0^\varepsilon(y):=\sup\{q^\varepsilon(y;\bar y,h(\bar y)):\bar y\in(0,y_*-M\varepsilon^{1/3})\}$. Now we will prove that $q_0^\varepsilon(y)=\tilde q_0^\varepsilon(y)$.

First notice that since $(h^\varepsilon)'(y)<0=f^\varepsilon(y,h^\varepsilon(y))$, $h^\varepsilon(y)$ is a subsolution of Equation~\eqref{eq:odeq}. Therefore from $(q_0^\varepsilon)'(0^+)=0$ and $(h^\varepsilon)'(0^+)<0$, we get that $q_0^\varepsilon(y)>h^\varepsilon(y)$ on $(0,y_*-M\varepsilon^{1/3})$. In addition, for any $\bar y$, $q(y;\bar y,h^\varepsilon(\bar y))<h^\varepsilon(y)$ on $(0,\bar y)$. This shows that $q_0^\varepsilon(y)\geq\tilde q_0^\varepsilon(y)$. Let's assume by contradiction that $\tilde q_0^\varepsilon(0^+)<q_0^\varepsilon(0^+)=h^\varepsilon(0^+)$. This means that $\tilde q_0^\varepsilon(\bar y)<h^\varepsilon(\bar y)$ for some $\bar y>0$, that is $\tilde q_0^\varepsilon(\bar y)<q_0^\varepsilon(\bar y;\bar y,h^\varepsilon(\bar y))$, contradicting the maximality of $\tilde q_0^\varepsilon(y)$. This proves that $q_0^\varepsilon(y)=\sup\{q^\varepsilon(y;\bar y,h(\bar y)):\bar y\in(0,y_*-M\varepsilon^{1/3})\}$.

Now, let $r(z;\bar z, \bar r)$ be the solution to equation (\ref{leftr}-\ref{rightr}) with initial condition $(\bar z,\bar r)$. By continuity of the solutions to \eqref{eq:odeq} with respect to parameters, $\varepsilon^{-1}q_0^\varepsilon(\varepsilon^{1/3}z+y_*)$ converges to $r_l(z):=\sup\{r(z;\bar z,k(\bar z)):\bar z\in(-\infty,-M)\}$.

In the proof of Lemma~\ref{lemma_riccati} we show that $r_l(z)$ as defined above satisfies
\begin{equation*}
\lim_{z\rightarrow-\infty}\frac{r_l(z)}{-\sqrt{2K\gamma\sigma^2}z}=1.
\end{equation*}
Recall that $q_1^\varepsilon$ is the unique solution  to equation~\eqref{eq:odeq} with limit $q^\varepsilon(1^-)=b_1(\varepsilon)$. With the same arguments used for $q_0^\varepsilon$, we get that $\varepsilon^{-1}q_1^\varepsilon(\varepsilon^{1/3}z+y_*)$ converges to the unique solution $r_r(z)$ to equation~\eqref{rightr} that satisfies
\begin{equation*}
\lim_{z\rightarrow+\infty}\frac{r_r(z)}{-\sqrt{2K\gamma\sigma^2}z}=1.
\end{equation*}
In the proof of Theorem~\ref{mainresult} we show that, for the optimal value of $\beta$, $q^\varepsilon(y)=q^\varepsilon_0(y)=q^\varepsilon_1(y)$. Thus, the function $\varepsilon^{-1}q^\varepsilon(\varepsilon^{1/3}z+y_*)$ converges to the unique solution of (\ref{leftr}-\ref{rightr}) with growth conditions \eqref{rlminusinfty}. In particular, given the uniqueness of $l_*$ in Proposition~\ref{prop_asymptotics}, $\liminf_{\varepsilon\rightarrow0}\frac{c(\varepsilon)}{\varepsilon^{2/3}}=\limsup_{\varepsilon\rightarrow0}\frac{c(\varepsilon)}{\varepsilon^{2/3}}=l=l_*$ as claimed.
\end{proof}

\bibliographystyle{abbrv}
\bibliography{priceimpact}

\begin{thebibliography}{10}

\bibitem{almgren.chriss.01}
R.~F. Almgren and N.~Chriss.
\newblock Optimal execution of portfolio transactions.
\newblock {\em J. Risk}, 3:5--40, 2001.

\bibitem{almgren.li.11}
R.~F. Almgren and T.~M. Li.
\newblock Option hedging with smooth market impact.
\newblock {\em Market Microstructure Liq.}, 2(1):1650002, 2016.

\bibitem{altarovici.al.13}
A.~Altarovici, J.~Muhle-Karbe, and H.~M. Soner.
\newblock Asymptotics for fixed transaction costs.
\newblock {\em Finance Stoch.}, 19(2):702--723, 2015.

\bibitem{bertsimas.lo.98}
D.~Bertsimas and A.~W. Lo.
\newblock Optimal control of execution costs.
\newblock {\em J. Financ. Markets}, 1(1):1--50, 1998.

\bibitem{bichuch.12}
M.~Bichuch.
\newblock Asymptotic analysis for optimal investment in finite time with
  transaction costs.
\newblock {\em SIAM J. Financ. Math.}, 3(1):433--458, 2012.

\bibitem{chen.dai.13}
X.~Chen and M.~Dai.
\newblock Asymptotics for {M}erton problem with small capital gain tax and
  interest rate.
\newblock Preprint, 2013.

\bibitem{dufresne.al.12}
P.~Collin-Dufresne, K.~Daniel, C.~Moallemi, and M.~Saglam.
\newblock Strategic asset allocation with predictable returns and transaction
  costs.
\newblock Preprint, 2013.

\bibitem{constantinides.86}
G.~Constantinides.
\newblock {Capital market equilibrium with transaction costs}.
\newblock {\em J. Polit. Econ.}, 94(4):842--862, 1986.

\bibitem{dai.yi.09}
M.~Dai and F.~Yi.
\newblock Finite-horizon optimal investment with transaction costs: a parabolic
  double obstacle problem.
\newblock {\em J. Diff. Eq.}, 246(4):1445--1469, 2009.

\bibitem{davis.norman.90}
M.~H.~A. Davis and A.~R. Norman.
\newblock Portfolio selection with transaction costs.
\newblock {\em Math. Oper. Res.}, 15(4):676--713, 1990.

\bibitem{dumas.luciano.91}
B.~Dumas and E.~Luciano.
\newblock {An exact solution to a dynamic portfolio choice problem under
  transaction costs}.
\newblock {\em J. Finance}, 46(2):577--595, 1991.

\bibitem{garleanu.pedersen.13a}
N.~Garleanu and L.~H. Pedersen.
\newblock Dynamic trading with predictable returns and transaction costs.
\newblock {\em J. Finance}, 68(6):2309--2340, 2013.

\bibitem{garleanu.pedersen.13b}
N.~Garleanu and L.~H. Pedersen.
\newblock Dynamic portfolio choice with frictions.
\newblock {\em J. Econ. Theory}, 164:487--516, 2016.

\bibitem{gerhold.al.14}
S.~Gerhold, P.~Guasoni, J.~Muhle-Karbe, and W.~Schachermayer.
\newblock {Transaction costs, trading volume, and the liquidity premium}.
\newblock {\em Finance Stoch.}, 18(1):1--37, 2014.

\bibitem{grossman.vila.92}
S.~Grossman and J.~Vila.
\newblock {Optimal dynamic trading with leverage constraints}.
\newblock {\em J. Financ. Quant. Anal.}, 27(2):151--168, 1992.

\bibitem{grossman.zhou.93}
S.~J. Grossman and Z.~Zhou.
\newblock {Optimal investment strategies for controlling drawdowns}.
\newblock {\em Math. Finance}, 3(3):241--276, 1993.

\bibitem{guasoni.robertson.12}
P.~Guasoni and S.~Robertson.
\newblock Portfolios and risk premia for the long run.
\newblock {\em Ann. Appl. Probab.}, 22(1):239--284, 2012.

\bibitem{guasoni.weber.nonlinear}
P.~Guasoni and M.~H. Weber.
\newblock Nonlinear price impact and portfolio choice.
\newblock {\em Preprint}, 2015.

\bibitem{guasoni.weber.13}
P.~Guasoni and M.~H. Weber.
\newblock Dynamic trading volume.
\newblock {\em Math. Finance}, 27(2):313--349, 2017.

\bibitem{janecek.shreve.04}
K.~{Jane\v cek} and S.~E. Shreve.
\newblock Asymptotic analysis for optimal investment and consumption with
  transaction costs.
\newblock {\em Finance Stoch.}, 8(2):181--206, 2004.

\bibitem{kallsen.li.13}
J.~Kallsen and S.~Li.
\newblock Portfolio optimization under small transaction costs: a convex
  duality approach.
\newblock Preprint, 2013.

\bibitem{kallsen.muhlekarbe.13}
J.~Kallsen and J.~Muhle-Karbe.
\newblock The general structure of optimal investment and consumption with
  small transaction costs.
\newblock {\em Math. Finance}, 27(3):659--703, 2017.

\bibitem{kamke.77}
E.~Kamke.
\newblock {\em Differentialgleichungen. L\"osungsmethoden und L\"osungen. I:
  Gew\"ohnliche Differentialgleichungen}.
\newblock B. G. Teubner, Stuttgart, ninth edition, 1977.

\bibitem{karatzas.shreve.91}
I.~Karatzas and S.~E. Shreve.
\newblock {\em Brownian motion and stochastic calculus}.
\newblock Springer, New York, second edition, 1991.

\bibitem{korn.98}
R.~Korn.
\newblock Portfolio optimisation with strictly positive transaction costs and
  impulse control.
\newblock {\em Finance Stoch.}, 2(2):85--114, 1998.

\bibitem{madhavan.00}
A.~Madhavan.
\newblock Market microstructure: a survey.
\newblock {\em Journal Financ. Markets}, 3(3):205--258, 2000.

\bibitem{magill.constantinides.76}
M.~J.~P. Magill and G.~M. Constantinides.
\newblock Portfolio selection with transactions costs.
\newblock {\em J. Econ. Theory}, 13(2):245--263, 1976.

\bibitem{martin.12}
R.~Martin.
\newblock Optimal trading under proportional transaction costs.
\newblock {\em RISK}, August:54--59, 2014.

\bibitem{merton.69}
R.~C. Merton.
\newblock {Lifetime portfolio selection under uncertainty: the continuous-time
  case}.
\newblock {\em Rev. Econ. Statist.}, 51(3):247--257, 1969.

\bibitem{moreau.al.14}
L.~Moreau, J.~Muhle-Karbe, and H.~M. Soner.
\newblock Trading with small price impact.
\newblock {\em Math. Finance}, 27(2):350--400, 2017.

\bibitem{obizhaeva.wang.13}
A.~A. Obizhaeva and J.~Wang.
\newblock Optimal trading strategy and supply/demand dynamics.
\newblock {\em J. Financ. Markets}, 16(1):1--32, 2013.

\bibitem{possamai.al.13}
D.~Possamai, H.~M. Soner, and N.~Touzi.
\newblock Homogenization and asymptotics for small transaction costs: the
  multidimensional case.
\newblock {\em Comm. Part. Diff. Eq.}, 40(11):2005--2046, 2015.

\bibitem{schied.schoeneborn.09}
A.~Schied and T.~Sch\"oneborn.
\newblock Risk aversion and the dynamics of optimal liquidation strategies in
  illiquid markets.
\newblock {\em Finance Stoch.}, 13(2):181--204, 2009.

\bibitem{shreve.soner.94}
S.~E. Shreve and H.~M. Soner.
\newblock Optimal investment and consumption with transaction costs.
\newblock {\em Ann. Appl. Probab.}, 4(3):609--692, 1994.

\bibitem{slater.60}
L.~J. Slater.
\newblock {\em Confluent hypergeometric functions}.
\newblock Cambridge University Press, New York, 1960.

\bibitem{soner.touzi.13}
H.~M. Soner and N.~Touzi.
\newblock Homogenization and asymptotics for small transaction costs.
\newblock {\em SIAM J. Control Optim.}, 51(4):2893--2921, 2013.

\bibitem{whittaker.watson.96}
E.~Whittaker and G.~N. Watson.
\newblock {\em A course of modern analysis}.
\newblock Cambridge University Press, 1996.

\end{thebibliography}

\end{document}